\documentclass[11pt,english]{article}
\usepackage[latin9]{inputenc}
\usepackage[a4paper]{geometry}
 \usepackage [Sonny]{fncychap}
\usepackage{fancyhdr}

\fancypagestyle{normal}{
\fancyhf{}
\fancyhead[EL]{\thepage}
\fancyhead[ER]{\slshape \nouppercase {\scshape{\chaptername \hspace{0.1cm} \thechapter}}}
\fancyhead[OL]{\slshape \nouppercase{\scshape{\rightmark}}}
\fancyhead[OR]{\thepage}
 }

\fancypagestyle{ohne}{
\fancyhf{}
}

\fancypagestyle{bib}{
\fancyhf{}
\fancyhead[EL]{\thepage}
\fancyhead[ER]{\slshape \nouppercase {\scshape{\leftmark}}}
\fancyhead[OL]{\slshape \nouppercase {\scshape{\leftmark}}}
\fancyhead[OR]{\thepage}
}

\usepackage{color}
\usepackage{amsmath}
\usepackage{amsthm}
\usepackage{amssymb}
\usepackage{amsmath}
\usepackage{mathcomp}
\usepackage{dsfont}
\usepackage[toc,page]{appendix}
\usepackage{graphicx} 
\usepackage{subfigure}
\usepackage{enumerate}
\usepackage{fancyhdr}

\usepackage{hyperref}

\usepackage{mathtools}

\newtheorem{definition}{Definition}[section]
\newtheorem{lemma}[definition]{Lemma}

\newtheorem{proposition}[definition]{Proposition}
\newtheorem{theorem}[definition]{Theorem}

\numberwithin{equation}{section}

\newcommand{\rk}{\text{r}}
\newcommand{\pk}{\text{p}}
\newcommand{\kk}{\text{k}}
\newcommand{\rd}{\text{d}}
\newcommand{\psit}{\Psi_{N,t}}

\newcommand{\qN}{q_{N,t}}
\newcommand{\vN}{\varphi_{N,t}}

\begin{document}

\title{Central limit theorem for Bose gases interacting through singular potentials}

\author{Simone Rademacher \\
\\
Institute of Mathematics, University of Zurich, \\
Winterthurerstrasse 190, 8057 Zurich, 
Switzerland, \\
current adress: IST Austria, Am Campus 1, 3400 Klosterneuburg, Austria}

\maketitle

\maketitle

\begin{abstract}
We consider a system of $N$ bosons in the limit $N \rightarrow \infty$, interacting through singular potentials. For initial data exhibiting Bose-Einstein condensation, the many-body time evolution is well approximated through a quadratic fluctuation dynamics around a cubic non-linear Schr\"odinger equation of the condensate wave function. We show that these fluctuations satisfy a (multi-variate) central limit theorem. 
\end{abstract}

\section{Introduction}

We consider a system of $N$ bosons with Hamilton operator 
\begin{align}
\label{eq:HNbeta}
H_N = \sum_{j=1}^N - \Delta_{x_j} + \frac{1}{N} \sum_{1 \leq j < k \leq N} V_N (x_j-x_k)
\end{align} 
acting on $L^2_s( \mathbb{R}^{3N}) $, the subspace of $L^2( \mathbb{R}^{3N})$ consisting of functions which are symmetric with respect to permutations. The $N$-dependent two-body interaction potential is given through
\begin{align*}
V_N(x) = N^{3 \beta} V( N^\beta x).
\end{align*}
In the following, we assume $V \geq 0$ to be smooth, spherically symmetric and compactly supported. For $\beta =0$, the Hamiltonian \eqref{eq:HNbeta} describes the mean-field regime characterized by a large number of weak collisions, whereas for $\beta >1/3$ the collisions of the particles are rare but strong. In the Gross-Pitaevskii regime ($\beta =1$), pair correlations play a crucial role. Here, we study intermediate regimes $\beta \in (0,1)$ in the limit $N \rightarrow \infty$ where the particles interact through singular potentials.

The time evolution is governed by the Schr\"odinger equation
\begin{align}
\label{eq:Schroe}
i \partial_t \psi_{N,t} = H_N \psi_{N,t}.
\end{align} 
For $\beta =0$ (mean-field regime), the solution of \eqref{eq:Schroe} can be approximated by products of solutions of the Hartree equation
\begin{align*}
i \partial_t \varphi_t = - \Delta \varphi_t + \left( V* | \varphi_t|^2 \right) \varphi_t
\end{align*}
with initial data $\varphi_0 \in L^2( \mathbb{R}^3)$. See for example \cite{AGT,AFP,AN,AH,BGM,BCS,ES,EY,FKP,FKS,KSS,KP,Spohn_Kinetic}. For $0 < \beta \leq 1$, on the other hand, the solution $\psi_{N,t}$ of \eqref{eq:Schroe} can be approximated by the non-linear Schr\"odinger equation 
\begin{align}
\label{eq:nonlinear Schroe}
i \partial_t \varphi_t = - \Delta \varphi_t + \sigma \varphi_t 
\end{align}
with $\sigma = \widehat{V}(0)$ if $\beta <1$ and $\sigma = 8 \pi \mathfrak{a}_0$ if $\beta =1$ (Gross-Pitaevskii regime). Hereafter, $\mathfrak{a}_0$ denotes the scattering length associated with the potential $V$ defined through the solution of the zero-energy scattering equation
\begin{align}
\left[ - \Delta +\frac{1}{2} V \right] f = 0
\end{align}
with boundary condition $f(x) \rightarrow 1$ as $|x| \rightarrow \infty$. Then, outside the support of $V$, the solution $f$ is given through
\begin{align}
f(x) =1 - \mathfrak{a}_0/|x|,
\end{align}
where $\mathfrak{a}_0$ is defined as the scattering length of the potential $V$. In \cite{ESY_beta_einhalb,ESY_GP_rig,ESY_GP} it is shown that if the one-particle reduced density $\gamma_N$ associated with $\psi_N$ satisfies
\begin{align*}
\gamma_N \rightarrow | \varphi_0 \rangle \langle \varphi_0 |
\end{align*}
in the trace norm topology and
\begin{align}
\label{eq:energycond}
\langle \psi_N, H_N \psi_N \rangle \leq CN,
\end{align}
then the one particle reduced density $\gamma_{N,t}$ associated with the solution $\psi_{N,t}$ of \eqref{eq:Schroe} obeys
\begin{align}
\label{eq:BEC time}
\gamma_{N,t} \rightarrow | \varphi_{t} \rangle \langle \varphi_t |
\end{align}
where $\varphi_t$ denotes the solution of \eqref{eq:nonlinear Schroe}. In fact, in \cite{ESY_beta_einhalb} considering the case $\beta <1$, the energy condition \eqref{eq:energycond} on the initial data is not needed. For more results in the Gross-Pitaevskii regime see \cite{BOS,BS,CH_GP,P_GP_1,P_extPot}. An overview on the derivation of the non-linear Schr\"odinger equation from many-body quantum dynamics is given in \cite{BPS,Golse,Schlein}. \\

\subsection{Norm approximation}

Besides the convergence of the one-particle reduced density $\gamma_{N,t}$ associated with $\psi_{N,t}$, the norm approximation of $\psi_{N,t}$ has been studied for different settings of $\beta \in (0,1)$ in \cite{BNNS,GM,NM_betaeindrittel,NM_betaeinhalb}. Our result is based on the norm approximation obtained in \cite{BNNS} covering $\beta <1$ whose ideas we explain in the following. \\


\textit{Truncated Fock space.}  As first step towards the norm approximation in \cite{BNNS}, the contribution of the Bose-Einstein condensate is factored out. This is realized through the unitary $\mathcal{U}_{\varphi_{N,t}} : L_s \left( \mathbb{R}^{3N} \right) \rightarrow \mathcal{F}_{\perp \varphi_{N,t}}^{\leq N}$. It maps the $N$-particle sector of the bosonic Fock space
$$
\mathcal{F} = \bigoplus_{n \geq 0} L_s \left( \mathbb{R}^{3n} \right)
$$ 
into the truncated Fock space 
\begin{align*}
 \mathcal{F}_{\perp \varphi_{N,t}}^{\leq N} = \bigoplus_{n = 0}^N L^2_{\perp \varphi_{N,t}} \left( \mathbb{R}^{3} \right)^{\otimes_s n}          
\end{align*}
defined over the orthogonal complement $L^2_{\perp \varphi_{N,t}}\left( \mathbb{R}^3\right)$ of the subspace of $L^2( \mathbb{R}^3)$ spanned by the condensate wave function $\varphi_{N,t}$.  This unitary has first been used in \cite{LNSS} in the mean-field regime. Its definition is based on the observation that every $\psi_N \in L^2_s ( \mathbb{R}^{3N}) $ has a unique decomposition
\begin{align*}
\psi_N = \sum_{n=0}^N \alpha^{(n)} \otimes_s \varphi_{N,t}^{N-n}
 \end{align*}
where $\alpha^{(n)} \in L^2_{\perp \varphi_{N,t}} ( \mathbb{R}^3)^{\otimes_s n}$ for all $n=1, \cdots, N$. Then,
\begin{align*}
\mathcal{U}_{ \varphi_{N,t}} \psi_N = \lbrace \alpha^{(0)}, \alpha^{(1)}, \cdots, \alpha^{(n)} \rbrace .
\end{align*}
This unitary satisfies the following properties proven in \cite{LNSS}
\begin{align}
\label{eq:propU}
\mathcal{U}_{\varphi_{N,t}} a^*( \varphi_{N,t} ) a( \varphi_{N,t} ) \mathcal{U}_{ \varphi_{N,t}}^* =& N - \mathcal{N}_+(t)\notag \\
\mathcal{U}_{ \varphi_{N,t}} a^*( \varphi_{N,t} ) a( f ) \mathcal{U}_{ \varphi_{N,t}}^* =& \sqrt{N - \mathcal{N}_+(t)} a(f) \notag\\
\mathcal{U}_{ \varphi_{N,t}} a^*( f ) a( \varphi_{N,t} ) \mathcal{U}_{ \varphi_{N,t}}^* =& a^*(f) \sqrt{N - \mathcal{N}_+(t)}  \notag \\
\mathcal{U}_{ \varphi_{N,t}} a^*( f ) a( g ) \mathcal{U}_{ \varphi_{N,t}}^* =& a^*(f) a(g)
\end{align}
for all $f,g \in L^2_{\perp \varphi_{N,t}} ( \mathbb{R}^3)$. Here $a^*(f), a(f)$ denote the standard creation and annihilation operators on the bosonic Fock space $\mathcal{F}$. On the truncated Fock space, we define modified creation and annihilation operators
\begin{align}
\label{eq:modified c/a}
 b^*(f) = a^*(f) \sqrt{\frac{N-\mathcal{N}_+(t)}{N}}, \hspace{0.3cm}b(f)= \sqrt{\frac{N-\mathcal{N}_+(t)}{N}} a(f). 
\end{align}
The modified creation operator $b^*(f)$ excites one particle from the condensate into its complement while $b(f)$ annihilates an excitation into the condensate. We define the vector $ \xi_{N,t} :=  \mathcal{U}_{\varphi_{N,t}}  \psi_{N,t}$ representing the fluctuation outside the condensate and observe 
\begin{align}
\label{eq:def_Xi}
i \partial_t \xi_{N,t} = \mathcal{L}_{N,t} \;  \xi_{N,t}, \hspace{0.3cm } \text{with} \hspace{0.3cm} \mathcal{L}_{N,t} = \mathcal{U}_{\varphi_{N,t}} H_N \mathcal{U}_{\varphi_{N,t}}^* + \left( i \partial_{t} \mathcal{U}_{\varphi_{N,t}} \right) \mathcal{U}_{\varphi_{N,t}}^*
\end{align}
with initial data $\xi_{N,0} = \mathcal{U}_{\varphi_{N,0}} \psi_{N,0}$. \\


\textit{From the truncated Fock space to the bosonic Fock space.} We approximate the generator $\mathcal{L}_{N,t}$ acting on the truncated Fock space only with a modified generator $\widetilde{\mathcal{L}}_{N,t}$ defined on the whole bosonic Fock space. We consider regimes with a small number of excitations $\mathcal{N}_+(t)$. For this reason, we realize the approximation of  $\mathcal{L}_{N,t}$ through $\widetilde{\mathcal{L}}_{N,t}$ by replacing $\sqrt{N-\mathcal{N}_+(t)}$ with $\sqrt{N} G_M(\mathcal{N}_+(t) /N)$, where
\begin{align*}
G_M(\tau) := \sum_{n=0}^M \frac{(2n)!}{(n!)^2 4^n (1-2n)} t^n 
\end{align*}
is the $M$-th Taylor polynom of $\sqrt{1-\tau}$ expanded at the point $\tau_0 =0$. For a precise definition see \cite[eq. (54)]{BNNS}. 
%
%
 \\

\textit{Correlation structure through Bogoliubov transformation.}  In the intermediate regime correlations are important (at least if $\beta >1/2$). For their implementation, we consider for fixed $\ell >0$ the ground state of the scattering equation 
\begin{align}
\label{eq:scattering_eq}
\left[ - \Delta + \frac{1}{2N} V_N \right] f_N = \lambda_N f_N
\end{align}
with Neumann boundary conditions on the ball $B_\ell (0) $. We fix $f_N(x) = 1$ for all $|x| =\ell$ and extend $f_N$ to $\mathbb{R}^3$ by setting $f_N(x) = 1$ for all $|x| \geq \ell$.

In \cite{BCS}, the non linear Schr\"odinger equation \eqref{eq:nonlinear Schroe} is replaced by the $N$-dependent Hartree equation 
\begin{align}
\label{eq:modified Hartree}
i \partial_t \varphi_{N,t} = - \Delta \varphi_{N,t} + ( V_Nf_N * | \varphi_{N,t}|^2) \varphi_{N,t}
\end{align}
with initial data $\varphi_{N,0} = \varphi_0$ (the condensate wave function at time $t=0$) to approximate the time evolved condensate wave function. The well-posedness of \eqref{eq:modified Hartree} is shown in \cite[Appendix B]{BCS}.

The correlation structure is implemented through the Bogoliubov transformation 
\begin{align}
\label{eq:def_bogo}
T_{N,t} = \exp \left( \frac{1}{2} \int dxdy \; \left[ \eta_{N,t} (x,y) a_xa_y - h.c. \right] \right) 
\end{align}
where $\eta_{N,t}$ denotes the Hilbert-Schmidt operator with integral kernel
\begin{align*}
\eta_{N,t} (x;y ) = - (q_{N,t} \otimes q_{N,t} )N \omega_N ( x-y) \varphi_{N,t}^2(( x+y)/2).
\end{align*}
Here, $\omega_N = 1- f_N$ and $\varphi_{N,t}$ are as defined in \eqref{eq:scattering_eq} resp. \eqref{eq:modified Hartree} and $q_{N,t} = 1- | \varphi_{N,t} \rangle \langle \varphi_{N,t} |$. The Bogoliubov transformation acts on the creation and annihilation operators as
\begin{align}
\label{eq:action_bogo}
T_{N,t} \; a(f) T_{N,t}^* =& a \left( \cosh_{\eta_{N,t}}  (f) \right) + a^* \left( \sinh_{ \eta_{N,t}} \overline{f})\right) \notag\\
T_{N,t} \; a^*(f) T_{N,t}^* =& a^* \left( \cosh_{\eta_{N,t}} (f) \right) + a \left( \sinh_{\eta_{N,t}} ( \overline{f})\right)
\end{align}
for all $f \in L^2( \mathbb{R}^3)$.  The operators $\sinh_{\eta_{N,t}}$ and $\cosh_{\eta_{N,t}}$ are defined through the absolutely convergent series of products of the operator $\eta_{N,t}$
\begin{align}
\label{eq:def_cosh}
\cosh_{\eta_{N,t}} = \sum_{n \geq 0} \frac{1}{(2n)!} \left( \eta_{N,t} \overline{\eta}_{N,t} \right)^n, \hspace{0.3cm} \sinh_{\eta_{N,t}} = \sum_{n \geq 0} \frac{1}{(2n+1)!} \left( \eta_{N,t} \overline{\eta}_{N,t}\right)^n \eta_{N,t} . 
\end{align}
Let $\mathcal{G}_{N,t}$ be the generator given through
\begin{align}
\label{eq:def_GN}
\mathcal{G}_{N,t} = \left( i \partial_t T_{N,t}\right) T_{N,t}^* + T_{N,t} \widetilde{\mathcal{L}}_{N,t} T_{N,t}^* .
\end{align}
In fact, the special choice of \eqref{eq:scattering_eq} and \eqref{eq:modified Hartree} allow crucial cancellations in the generator $\mathcal{G}_{N,t}$. Note that  $\mathcal{G}_{N,t}$ consists of terms which are quadratic in creation and annihilation operators and of terms of higher order. Nevertheless, in \cite[Lemma 5]{BNNS} it is shown that $\mathcal{G}_{N,t}$ can be approximated through the generator $\mathcal{G}_{2,N,t}$ containing quadratic terms only. \\


\textit{Limiting quadratic dynamics.} We are interested in the limit $N\rightarrow \infty$ of $\mathcal{G}_{2,N,t}$ defined in \eqref{eq:def_GN}. In order to replace the Bogoliubov transformation $T_{N,t}$ defined in \eqref{eq:def_bogo} with a limiting one, we define the limiting kernel $\omega_\infty$ of $\omega_N$ through
\begin{align}
\label{eq:def_omega_infty}
\omega_\infty (x) = \frac{\mathfrak{b}_0}{8 \pi} \left[ \frac{1}{|x|} - \frac{3}{2\ell} + \frac{x^2}{3 \ell^3}\right]
\end{align}
for $|x| \leq \ell$ and $\omega_\infty (x) =0$ otherwise. Here, we used the notation $\mathfrak{b}_0 = \int dx \; V(x) $. 

%
%

Furthermore, the solution $\varphi_{N,t}$ of the modified Hartree equation \eqref{eq:modified Hartree} with initial data $\varphi_0 \in H^4 ( \mathbb{R}^3)$ can be approximated with the solution $\varphi_t$ of \eqref{eq:nonlinear Schroe} with $\sigma= \widehat{V}(0)$ and with initial data $\varphi_0$. To be more precise, \cite[Proposition B.1]{BCS} shows that there exists a constant $C>0$ (depending on $\| \varphi_0 \|_{H^4}$) such that 
\begin{align}
\label{eq:approx phi}
\| \varphi_t - \varphi_{N,t} \|_2 \leq CN^{-\gamma} \exp \left( C \exp \left( C |t| \right) \right)
\end{align}
with $\gamma = \min \lbrace \beta, 1-\beta \rbrace$. Standard arguments (see for example \cite[Proposition B.1]{BCS})  imply, that there exists a constant $C>0$ such that  
\begin{align}
\label{eq:regularity phi}
\| \varphi_t \|_2 \leq C, \hspace{0.3cm} \| \varphi_t \|_\infty \leq C \exp ( C|t|)), \hspace{0.3cm} \| \varphi_t \|_{H^n} \leq C \exp ( C |t|))
\end{align}
for all $n \in \mathbb{N}$.
The approximations \eqref{eq:def_omega_infty} and \eqref{eq:approx phi} lead to a limiting kernel 
\begin{align}
\label{eq:def eta}
\eta_t (x;y) = - \left( q_t \otimes q_t\right) \omega_\infty (x-y) \varphi_t^2 ( (x+y)/2).
\end{align}
We define the limiting Bogoliubov transformation 
\begin{align}
\label{eq:def bogo limiting}
T_{t} = \exp \left( \frac{1}{2} \int dxdy \; \left[ \eta_{t} (x,y) a_xa_y - h.c. \right] \right).
\end{align}
In fact, \eqref{eq:def_omega_infty} and \eqref{eq:approx phi} yield that there exists a constant $C>0$ such that
\begin{align}
\label{eq:approx eta}
\| \eta_{N,t} - \eta_t \|_2 \leq CN^{-\gamma} \exp \left( C \exp \left( C |t| \right) \right)
\end{align}
where $\gamma = \min \lbrace \beta, 1- \beta \rbrace$.

In order to define the limiting dynamics we introduce some more notation. We use the shorthand notation $j_x( \cdot) = j(\cdot, x)$ for any $j \in L^2( \mathbb{R}^3 \times \mathbb{R}^3)$.  Furthermore, we decompose $\text{sh}_{\eta_t} = \eta_t + \rk_t $,  $\text{ch}_{\eta_t} = \mathds{1} + \pk_t$ and
\begin{align*}
\eta_{t} (x;y) 
=-\omega_\infty (x-y) \varphi_{t}^2 ((x+y)/2)  + \mu_t (x;y) 
= k_{t}(x;y) + \mu_{t}(x;y)
\end{align*}
for all $x,y \in \mathbb{R}^3$.

A slight modification of the arguments in \cite[Appendix C]{BCS} shows some properties of the kernels. For these, we consider initial data  $\varphi_0 \in H^4 ( \mathbb{R}^3)$ of \eqref{eq:nonlinear Schroe}. There exist a constant $C>0$ (depending only on $\| \varphi_{0}\|_{H^4( \mathbb{R}^3)}$ and on $V$) such that on one hand
\begin{align}
\label{eq:estimates_k,sh,p,r,mu}
\| \text{ch}_{\eta_t} \| \leq C, \hspace{0.3cm} \text{and} \hspace{0.3cm} \; \| k_{t} \|_2, \; \| \eta_{t} \|_2, \; \| \text{sh}_{\eta_t} \|_2,\;  \| \pk_t \|_2,\;  \| \rk_t \|_2, \; \| \mu_{t} \|_2 \leq C
\end{align}
where $\| \cdot \|$ denotes the operator norm. On the other hand, denoting with $\nabla_1 k_t $ and $\nabla_2 k_t$ the operator with the kernel $\nabla_x k_{t} (x;y)$ 
\begin{align}
\label{eq:estimates_nabla_eta}
 \| \partial_t \eta_{t} \|_2 \leq C e^{C|t|}, \hspace{0.2cm}  \max \left\lbrace \sup_{x} \int dz \; \vert \nabla_1 k_t (x;z) \vert , \; \sup_{y} \int dz \; \vert \nabla_1 k_t (z;y) \vert \right\rbrace \leq C.
\end{align}
Furthermore, let $\Delta_1 \rk_t $ resp. $\Delta_2 \rk_t $ be the operator having the kernel $\Delta_x \rk_t(x;y)$ resp. $\Delta_y \rk_t (x;y)$, then for all $i =1,2$
\begin{align}
\label{eq:estimates_delta_eta}
 \| \Delta_i \rk_t \|_2, \; \| \Delta_i \pk_t \|_2, \; \| \Delta_i \mu_{t} \|_2\leq C e^{C |t|} .
\end{align}
 In order to simplify notation, we write in the following $\text{sh}_{\eta_{t}} = \text{sh}, \; \text{ch}_{\eta_t} = \text{ch}$ resp. $\rk_t = \rk, \; \kk_t = \kk, \; \pk_t = \pk $.
 \\

\begin{definition} 
\label{def:U2}
We define the \textit{limiting dynamics} $\mathcal{U}_2( t;s)$ satisfying 
\begin{align}
\label{eq:def_U2}
i \partial_t \mathcal{U}_2 (t;s) = \mathcal{G}_2 (t) \mathcal{U}_2( t; s) \hspace{0.3cm} \text{and} \hspace{0.3cm} \mathcal{U}_2 (s;s) = \mathds{1}
\end{align}
where $\mathcal{G}_{2,t}$ is given by
\begin{align}
\label{eq:def_G2}
\mathcal{G}_{2,t} := \left( i \partial_t T_{t}\right) T_{t}^* + \mathcal{G}_{2,t}^\mathcal{V} + \mathcal{G}_{2,t}^\mathcal{K} + \mathcal{G}_{2,t}^{\lambda}
\end{align}
with
\begin{align}
\label{eq:def_G2V}
\mathcal{G}_{2,t}^{\mathcal{V}} =& \mathfrak{b}_0 \int dx \;  | \varphi_t (x)|^2 \;  \left[ a^* (\text{ch}_x ) a( \text{ch}_x)+a^*( \text{sh}_y) a( \text{sh}_x)  \right. \notag \\
& \hspace{4cm} \left. + a^*( \text{ch}_x) a^*(\text{sh}_x) +a(\text{ch}_x) a(\text{sh}_x)  \right]\notag \\
&+ \int dxdy \; K_{1,t}(x;y) \left[ a^*( \text{ch}_x) a(\text{ch}_y) + a^*( \text{sh}_x) a( \text{sh}_y) \right. \notag \\
& \left. \hspace{4cm}+ a^*( \text{ch}_x) a^* ( \text{sh}_y) + a( \text{ch}_y) a( \text{sh}_x)  \right] \notag\\
&+ \int dxdy \; K_{2,t}(x;y)\left[ a^*( \text{ch}_x) a( \text{sh}_y)+  a^*( \text{ch}_y) a( \text{sh}_x)\right.\notag \\
& \hspace{4cm} \left.+  a^* ( \text{ch}_x) a^*( \text{ch}_y) +  a(\text{sh}_x) a(\text{sh}_y) + h.c. \right] \notag\\
&+ \frac{1}{2} \left[ \| \varphi_t^2 \|_2^2 \; a^*( \varphi_t) a( \varphi_t) - 2a^*( \varphi_t) a( |\varphi_t|^2 \varphi_t) +h.c. \right] \notag\\
=& \sum_{i=1}^4 \mathcal{G}_{2,t}^{\mathcal{V},{(i)}} 
\end{align}
and
\begin{align}
\label{eq:def_G2L}
\mathcal{G}_{2,t}^{\lambda} = \frac{3\mathfrak{b}_0}{8 \pi \ell^3} \int dxdy \;   \chi \left( |x-y|\leq \ell \right) \varphi_{t}^2((x+y)/2) a_x^* a_y^* + h.c.
\end{align}
and
\begin{align}
\label{eq:def_G2K}
\mathcal{G}_{2,t}^{\mathcal{K}} - \mathcal{K}  =&  \int dx \; \left[ a_x^* a( - \Delta_x \pk_x)+  a^*( - \Delta_x \pk_x) a( \text{ch}_x) + a^*( \kk_x) a( - \Delta_x \rk_x) \right. \notag   \\
& \hspace{2cm} \left. + a^* ( \nabla_x \kk_x) a( \nabla_x \kk_x)  + a^*( - \Delta_x \rk_x) a( \rk_x) \right] \notag\\
&+  \int dx \; \left[ a_x^* a^*( - \Delta_x \mu_x) + a^*_x a^*( - \Delta_x \rk_x ) + a^*( - \Delta_x \pk_x) a^*( \text{sh}_x) \right. \notag \\
& \left. \hspace{2cm} + a( - \Delta_x \rk_x) a_x + a( - \Delta_x \mu_x) a_x  \right. \notag\\
&\hspace{2cm} + \left. a( \text{sh}_x) a( - \Delta_x \pk_x) + a^*(- \Delta_x \rk_x) a( \kk_x) \right] \notag\\
&+ \frac{1}{2} \int dxdy \; \omega_\infty (x-y)\varphi_t ((x+y)/2) \; \Delta \varphi_t ((x+y)/2)  a_x^*a_y^* + h.c. \notag\\
&+ \frac{1}{2} \int dxdy \; \omega_\infty (x-y) \nabla \varphi_t ((x+y)/2) \cdot \nabla \varphi_t ((x+y)/2)  a_x^*a_y^* + h.c. 
\end{align}
Here, we used the notation $\mathcal{K} = \int dx \; a_x^* \left( -  \Delta_x \right) a_x $ and  $K_{1,t} = q_{t} \widetilde{K}_{1,t} q_{t}$ and $  K_{2,t} = \left( q_{t} \otimes q_{t} \right) \widetilde{K}_{2,t}$ where $\widetilde{K}_{1,t}$ is the operator  with integral kernel
\begin{align*}
\widetilde{K}_{1,t}(x,y) = \mathfrak{b}_0 \varphi_{t}(x) \delta(x-y) \overline{\varphi_{t}(y)}
\end{align*}
and $K_{2,t}$ is the function given through
\begin{align*}
\widetilde{K}_{2,t} (x,y) = \mathfrak{b}_0 \varphi_{t} (x) \delta (x-y) \varphi_{t} (y) . 
\end{align*}
\end{definition}

Note that \eqref{eq:regularity phi} implies $K_{1,t}, K_{2,t} \in L^\infty( \mathbb{R}^6) \cap L^2( \mathbb{R}^6 )$ with norms uniform in $N$.  \\

\textit{Norm approximation.} We consider the solution $\psi_{N,t}$ of the Schr\"odinger equation \eqref{eq:Schroe} with initial data $\psi_{N,0} = U_{\varphi_0}^* \mathds{1}^{\leq N} T_{N,0}^* \Omega$. It is proven in \cite[Theorem 2]{BNNS} that for all $\alpha < \min \lbrace \beta/2, (1-\beta)/2 \rbrace$ there exists a constant $C>0$ such that  
\begin{align}
\label{eq:normapprox}
\| \mathcal{U}_{\varphi_{N,t}} \psi_{N,t} - e^{-i \int_0^t d\tau \; \eta_N (\tau ) } T_{N,t}^* \mathcal{U}_{2}(t;0) \Omega \|^2 \leq C N^{- \alpha} \exp( C \exp (C |t|))
\end{align}
for all $N$ sufficiently large and all $t \in \mathbb{R}$. \\


\subsection{Bogoliubov transformation}

The limiting dynamics $\mathcal{U}_{2} (t;s)$ defined in \eqref{eq:def_U2} is quadratic in creation and annihilation operators. As the following Proposition shows, it gives rise to a Bogoliubov transformation defined in the following. For this, we first define
\begin{align}
\label{eq:def_A}
A\left(f,g\right) = a^*\left(f\right) + a\left( \overline{g}\right) \hspace{0.3cm} \text{for} \hspace{0.3cm} f,g \in L^2\left( \mathbb{R}^3 \right).
\end{align}
On one hand
\begin{align}
\label{eq:property_A1}
A^*\left(f,g \right) = A\left(\overline{g}, \overline{f} \right) + A\left(\mathcal{J} (f,g) \right) \hspace{0.3cm} \text{with} \hspace{0.3cm} \mathcal{J} = \begin{pmatrix}
0 & J \\
J &0 
\end{pmatrix}.
\end{align}
Here,  $J: L^2\left(\mathbb{R}^3 \right) \rightarrow L^2 \left( \mathbb{R}^3 \right)$ denotes the anti-linear operator defined by $J f = \overline{f}$ for all $f \in L^2 \left( \mathbb{R}^3 \right)$. On the other hand, the commutation relations imply for $f_1,f_2, g_1,g_2 \in L^2( \mathbb{R}^3)$
\begin{align}
\label{eq:property_A2}
\left[ A\left( f_1,g_1 \right), A^* \left( f_2,g_2\right) \right] = \langle (f_1, g_1), S(f_2,g_2) \rangle_{L^2\left(\mathbb{R}^3 \right) \oplus L^2\left( \mathbb{R}^3 \right)} \hspace{0.3cm} \text{with} \hspace{0.3cm} S= \begin{pmatrix}
1 & 0 \\
0& -1 
\end{pmatrix}.
\end{align}
\begin{definition} A \textit{Bogoliubov transformation} is a linear map $\nu: L^2\left( \mathbb{R}^3 \right) \oplus L^2\left( \mathbb{R}^3 \right) \rightarrow L^2 \left( \mathbb{R}^3 \right) \oplus L^2 \left( \mathbb{R}^3 \right)$ preserving the relations \eqref{eq:property_A1} and \eqref{eq:property_A2}, i.e. $\nu^* S \nu = S$ and $ \mathcal{J} \nu =\nu \mathcal{J}$. 
\end{definition}

It turns out that a Bogoliubov transformation $\nu$ is of the form 
\begin{align*}
\nu= \begin{pmatrix}
U & JVJ \\
V & JUJ
\end{pmatrix}
\end{align*}
for linear operators $U,V : L^2 \left( \mathbb{R}^3 \right) \rightarrow L^2 \left( \mathbb{R}^3 \right)$ satisfying $U^* U - V^*V =1$ and $U^* JVJ - V^*JUJ =0$. 

The following Proposition is proven in Section \ref{subsec:proof_prop1}.

\begin{proposition}
\label{prop:bogo_U2}
Let $\mathcal{U}_{2}(t;s)$ be the dynamics defined in \eqref{eq:def_G2}.
For every $t,s \in \mathbb{R}$ there exists a bounded linear map 
\begin{align*}
\Theta (t;s) = L^2( \mathbb{R}^3) \oplus L^2( \mathbb{R}^3) \rightarrow L^2( \mathbb{R}^3) \oplus L^2( \mathbb{R}^3)
\end{align*} 
such that 
\begin{align*}
\mathcal{U}_{2}^*(t;s) A(f,g) \mathcal{U}_{2}(t;s) = A\left( \Theta(t;s) (f,g) \right)
\end{align*}
for all $f,g \in L^2( \mathbb{R}^3)$. The map $\Theta (t;s)$ satisfies
\begin{align}
\label{eq:bogo_prop}
\Theta (t;s)\mathcal{J} = \mathcal{J} \Theta (t;s), \hspace{0.5cm} S= \Theta(t;s)^* S \Theta(t;s)
\end{align}
where $\mathcal{J}$ and $S$ are defined in \eqref{eq:property_A1} resp. \eqref{eq:property_A2}. The Bogoliubov transformation $\Theta (t;s)$ can be written as
\begin{align*}
\Theta(t;s) = \begin{pmatrix}
U(t;s) & JV(t;s) J \\
V(t;s) & JU(t;s) J 
\end{pmatrix} 
\end{align*}
for bounded linear maps $U(t;s) , V(t;s) : L^2( \mathbb{R}^3) \rightarrow L^2( \mathbb{R}^3)$ satisfying
\begin{align}
\label{eq:propo_bogo_prop}
U^*( t;s) U(t;s) - V^*(t;s) V(t;s) = 1, \hspace{0.5cm} U^*(t;s) JV(t;s) J=V^*(t;s)JU(t;s) J. 
\end{align}
\end{proposition}

\subsection{Central limit theorem}

From a probabilistic point of view \eqref{eq:BEC time} implies a law of large numbers, in the sense that for a one-particle self-adjoint operator $O$ on $L^2( \mathbb{R}^3)$ and for every $\delta >0$
\begin{align}
\label{eq:lln}
\lim_{N \rightarrow \infty} \mathbb{P}_{\psi_{N,t}}  \left( \left\vert\frac{1}{N} \sum_{j=1}^N\left( O^{(j)} - \langle \varphi_t, O \varphi_t \rangle \right)  \right\vert > \delta \right) =0.
\end{align}
Here $O^{(j)}$ denotes the operator on $L^2( \mathbb{R}^{3N})$ acting as $O$ on the $j$-th particle and as identity elsewhere. The proof of \eqref{eq:lln} follows from Markov's inequality (see \cite{BSS}).  As a next step, we are interested in a central limit theorem. For this, we consider the rescaled random variable
\begin{align}
\label{eq:random var}
\mathcal{O}_{N,t} = \frac{1}{\sqrt{N}} \sum_{j=1}^N \left( O^{(j)} - \langle \varphi_{N,t}, O \varphi_{N,t} \rangle \right)
\end{align}
where $\varphi_{N,t}$ denotes the solution of \eqref{eq:modified Hartree} with initial data $\varphi_{N,0} = \varphi_0$.

We consider initial data $\psi_{N,0}$ of the form $\psi_{N,0} = \mathcal{U}_{\varphi_0}^* \mathds{1}^{\leq N} T_{N,0}^* \Omega $ exhibiting Bose-Einstein condensation \cite[Theorem 3]{BNNS}. As a consequence, such a initial data satisfy a law of large numbers in the sense of \eqref{eq:lln}. Moreover, such initial data obeys a central limit theorem in the sense that
\begin{align}
\label{eq:gauss0}
\mathbb{P}_{\psi_{N,0}} \left( \mathcal{O}_{N,0} \in [a;b] \right) \rightarrow  \mathbb{P} \left( G_0 \in [ a,b ] \right) \hspace{0.3cm} \text{as} \hspace{0.3cm} N \rightarrow \infty 
\end{align}
for every $- \infty < a < b < \infty $. Here, $G_0$ denotes the centered Gaussian random variable with variance $\| \sigma_0 \|_2^2$ where
\begin{align}
\label{eq:variance0}
\sigma_0 = \sinh_{\eta_0} \overline{q_0 O \varphi_0} + \cosh_{\eta_0} q_0 O \varphi_0 
\end{align}
following from Theorem \ref{thm:CLT} for time $t=0$. 

Note that initial data of the form $\psi_{N,0} = \mathcal{U}_{\varphi_0}^* \mathds{1}^{\leq N} T_{N,0}^* \Omega $ describe approximate ground states of trapped systems \cite{BCCS_betakleiner1}. In experiments such initial data are prepared by trapping particles through external fields and by cooling them down to extremely low temperatures so that the system essentially relaxes to its ground state. 

The validity of a central limit theorem for the ground state of trapped systems has already been adressed in \cite{RS_CLT}. To be more precise, \cite{RS_CLT} considers the ground state of \eqref{eq:HNbeta} for $\beta=1$, i.e. in the Gross-Pitaevskii regime. The ground state is known to exhibit Bose-Einstein condensation. It is proven that the ground state satisfies a central limit theorem. The arguments of the proof can be adapted to the intermediate regime $\beta < 1$ using the norm approximation for the ground state obtained in \cite{BCCS_betakleiner1}.  \\

Now, we consider the time evolution of the initial data $\psi_{N,0} = \mathcal{U}_{\varphi_0}^* \mathds{1}^{\leq N} T_{N,0}^* \Omega $  with respect to the Schr\"odinger equation \eqref{eq:Schroe} and show the validity of a (multi-variate) central limit theorem. 


\begin{theorem}
\label{thm:CLT}
Let $\beta \in (0,1)$ and assume $V$ to be radially symmetric, smooth, compactly supported and point-wise non-negative. Furthermore fix $\ell >0$ (independent of $N$). Let $\varphi_{t}$ denote the solution of \eqref{eq:nonlinear Schroe} and $\varphi_{N,t}$ the solution of \eqref{eq:modified Hartree} both with initial data $\varphi_{0} \in H^4( \mathbb{R}^3)$.  Moreover, we denote by $\psi_{N,t}$  the solution of the Schr\"odinger equation \eqref{eq:Schroe} with initial data $\psi_{N,0} = \mathcal{U}_{\varphi_{N,0}}^* \mathds{1}^{\leq N} T_{N,0}^* \Omega$ (where $\mathcal{U}_{\varphi_{N,0}}$ and $T_{N,0}$ are defined in \eqref{eq:propU} resp. \eqref{eq:def_bogo}). For $k \in \mathbb{N}$, let $O_1, \dots, O_k$ be bounded operators on $L^2( \mathbb{R}^3)$. We define $\nu_{j,t} \in L^2( \mathbb{R}^3)$ through
\begin{align}
\label{eq:covt}
\nu_{j,t}=&\left( U(t;0) \cosh_{\eta_t} + \overline{V}(t;0) \sinh_{\eta_t} \right) q_t O_j \varphi_t  \notag\\
&+ \left( U(t;0) \sinh_{\eta_t} + \overline{V}(t;0) \cosh_{\eta_t} \right) \overline{q_t O_j \varphi_t} 
\end{align}
where  the operators $U(t;0),V(t;0) \in L^2( \mathbb{R}^3 \times \mathbb{R}^3)$ are defined in Proposition \ref{prop:bogo_U2},  $q_t = 1- | \varphi_t \rangle \langle \varphi_t |$ and $\eta_t$ as defined in \eqref{eq:def eta}. 
 
Assume $\Sigma_t \in \mathbb{C}^{k \times k}$, given through
\begin{align*}
\left( \Sigma_t\right)_{i,j} = \begin{cases}
 \langle \nu_{i,t}, \nu_{j,t} \rangle & \text{for} \quad i<j \\
  \langle \nu_{j,t}, \nu_{i,t} \rangle & \text{otherwise} 
\end{cases}
\end{align*}
is invertible.

Furthermore, let  $g_1, \dots g_k \in L^1( \mathbb{R})$ with $\widehat{g}_i \in L^1( \mathbb{R}, (1+|s|)^4 ds)$ for all $i \in \lbrace 1, \dots k \rbrace$ and let $\mathcal{O}_{j,N,t}$ denote the random variable \eqref{eq:random var} associated to $O_j$ for all $j \in \lbrace 1, \dots, k \rbrace$.  For every $\alpha < \min \lbrace \beta/2, (1-\beta) /2 \rbrace$, there exists $C>0$ such that 
\begin{align*}
&\left\vert \mathbb{E}_{\psi_{N,t}} \right. \left[ g_1 ( \mathcal{O}_{1,N,t}) \dots g_k ( \mathcal{O}_{k,N,t} ) \right]  \\
&\hspace{2cm} - \left.  \frac{1}{\sqrt{(2 \pi)^k \det \Sigma}}  \int dx_1 \dots dx_k \; g_1(x_1) \dots g_k( x_k) \; e^{-\frac{1}{2} \sum_{i,j=1}^k \Sigma_{i,j}^{-1} x_j x_j} \right\vert \\
& \hspace{0.1cm}\leq C \exp ( \exp (C |t|)) N^{-\alpha} \prod_{j=1}^k \int d\tau \; | \widehat{g}_j (\tau) | \left( 1 +  N^{\alpha-\gamma} |\tau|^2 + N^{\alpha-1/2} | \tau|^3 + N^{\alpha -1} |\tau|^4 \right)
\end{align*}
where $\gamma = \min \lbrace \beta, (1-\beta) \rbrace $.
\end{theorem}


A similiar result has been established in \cite{BKS,BSS} for the mean-field regime characterized through weak interaction of the particles.  It is shown that fluctuations around the non-linear Hartree equation of bounded self-adjoint one-particle operators satisfy a (multi-variate) central limit theorem. We show that this result is true in the intermediate regime, where the interaction is singular, too. In particular, the correlation structure which becomes of importance in the intermediate regime does not affect the validity of a central limit theorem. Though, it affects the covariance matrix \eqref{eq:covt} through the Bogolioubiv transform $T_{t}$. \\

 Similarily as in \cite[Corollary 1.3 ]{BSS}, Theorem \ref{thm:CLT} implies a Berry-Ess\'{e}en type central limit theorem. To be more precise, we consider a bounded self-adjoint operator $O$ on $L^2( \mathbb{R}^3)$ and the random variable
 \begin{align*}
 \mathcal{O}_{N,t} = \frac{1}{\sqrt{N}} \sum_{i=1}^N \left( O^{(i)} - \langle \varphi_{N,t}, O \varphi_{N,t} \rangle \right) .
 \end{align*}
For every $\alpha < \min \lbrace \beta /2, (1-\beta) /2 \rbrace$ and $- \infty < a < b < \infty$, there exists a constant $C>0$ such that 
\begin{equation}\label{eq:cor}
\vert \mathbb{P}_{\psi_{N,t}} \left( \mathcal{O}_{N,t} \in \left[ a; b \right] \right) - \mathbb{P} \left( G_t \in \left[ a; b \right] \right) \vert \leq C N^{-\alpha/2}
\end{equation}
where $G_t$ is the centered Gaussian random variable with variance $\| \sigma_t \|_2^2$ and $\sigma_t \in L^2( \mathbb{R}^3)$ is defined through 
\begin{align}
\label{eq:variance}
\sigma_t = \left( U(t;0) \cosh_{\eta_t} + \overline{V}(t;0) \sinh_{\eta_t} \right) q_t O \varphi_t + \left( U(t;0) \sinh_{\eta_t} + \overline{V}(t;0) \cosh_{\eta_t} \right) \overline{q_t O \varphi_t } .
\end{align}
\vspace{0.3cm}

Note that Theorem \ref{thm:CLT} resp. \eqref{eq:cor} imply that fluctuations around the non-linear Hartree equation with singular interaction satisfy a (multi-variate) central limit theorem. Comparing with $\sigma_0$ from \eqref{eq:variance0}, the fluctuations enter in the variance $\sigma_t$ through the operators $U(t;0), V(t;0)$ as defined in Proposition \eqref{prop:bogo_U2} and the Bogoliubov transformation \eqref{eq:def bogo limiting}. \\

Moreover, note that the covariance matrix \eqref{eq:covt} resp. the variance \eqref{eq:variance} are completely determined by the Bogoliubov transform $T_{t}$ defined in \eqref{eq:def bogo limiting} and the quadratic fluctuation dynamics $\mathcal{U}_2 (t;0)$ defined in \eqref{eq:def_G2}. Theorem \ref{thm:CLT} resp. the properties \eqref{eq:propo_bogo_prop} of the operators $U(t;0), V(t;0)$ show that the solution of the Schr\"odinger equation \eqref{eq:Schroe} modulo the extraction of the condensate, is approximately a quasi-free state for quasi-free initial data. This observation coincides with results in \cite{GM,NM_betaeindrittel,NM_betaeinhalb}. \\


\section{Proof of Results}
 
\subsection{Preliminaries}

The proof of Theorem \ref{thm:CLT} is based on the norm approximation \eqref{eq:normapprox} from \cite{BNNS}. In the following we collect useful properties of the unitaries used therein. 

To this end, we define the more general quadratic dynamics $\mathcal{U}_{\mathrm{gen}}(t;s)$. 

\begin{definition}
\label{def:U_gen}
Let $\mathcal{U}_{\mathrm{gen}}(t;s)$ be the dynamics satisfying
\begin{align}
\label{eq:U2_allg}
i \partial_t\mathcal{U}_{\mathrm{gen}}(t;s) = \mathcal{G}_{\mathrm{gen},t} \mathcal{U}_{\mathrm{gen}}(t;s)
\end{align}
 where the generator $\mathcal{G}_{\mathrm{gen},t}$ is of the form
\begin{align}
\label{eq:G2_allg}
\mathcal{G}_{\mathrm{gen},t} =& \int dx \; \nabla_x a_x^* \nabla_x a_x + \int dxdy \; H_t^{(1)} (x;y) a_x^* a_y \notag \\
&+ \int dxdy \; \left( H_t^{(2)} (x,y) a_x^* a_y^* + \overline{H_t^{(2)}} (x,y) a_x a_y \right) + c
\end{align}
with
\begin{align}
\label{ass1}
\| H_t^{(1)} \| \leq C e^{C|t|}, \quad  \| H_t^{(2)} \|_2 \leq C e^{C |t|}
\end{align}
for constants $c,C>0$. 
\end{definition}

In the following, we prove the results for the dynamics $\mathcal{U}_{\mathrm{gen}}(t;s)$. As the next Lemma shows, the results then apply to $\mathcal{U}_2 (t;s)$, too.

\begin{lemma}
\label{lemma:U2}
The dynamics $\mathcal{U}_2(t;s)$ defined in Definition \ref{def:U2} is of the form of $\mathcal{U}_{\mathrm{gen}}(t;s)$ defined in Definition \ref{def:U_gen}. 
\end{lemma}

\begin{proof}

By the definition \eqref{eq:def_G2} of $\mathcal{G}_{2,t}$, we split
\begin{align}
\label{eq:bogo_G2_1}
\mathcal{G}_{2,t}  - \mathcal{K}= \left( i \partial_t T_{t} \right) T_{t}^* + \mathcal{G}^{\mathcal{V}}_{2,t} + \left(\mathcal{G}^{\mathcal{K}}_{2,t} - \mathcal{K} \right) + \mathcal{G}^{\lambda}_{2,t} 
\end{align}
and consider each of the summands separately. First, we consider $\mathcal{G}_{2,t}^\mathcal{V}$ defined in \eqref{eq:def_G2V}, which is again split into four terms. The first one, $\mathcal{G}_{2,t}^{\mathcal{V},1}$ of the r.h.s. of \eqref{eq:def_G2V} satisfies assumption \eqref{ass1} since on one hand 
\begin{align*}
\|\text{ch}_{\eta_t}  | \varphi_t|^2 \text{ch}_{\eta_t} \|_2 \leq \| \varphi_t \|_4^2 \leq C e^{C |t|}, \quad  \|\text{sh}_{\eta_t} \varphi_t |^2 \text{sh}_{\eta_t} \|_2 \leq \| \varphi_t \|_\infty^2  \| \text{sh}_{\eta_t} \|_2^2  \leq C e^{C |t|}
\end{align*}
and 
\begin{align*}
\| \text{ch}_{\eta_t} | \varphi_t|^2 \text{sh}_{\eta_t} \|_2 \leq C \| \varphi_t \|_2  \| \varphi_t \|_\infty \| \text{sh}_{\eta_t} \|_2 \leq C e^{C |t|}
\end{align*}
following from \eqref{eq:regularity phi} and \eqref{eq:estimates_k,sh,p,r,mu}. For the same reasons, the second term $\mathcal{G}_{2,t}^{\mathcal{V},2}$ of the r.h.s. of \eqref{eq:def_G2V} satisfies assumption \eqref{ass1}, too. For the third term $\mathcal{G}_{2,t}^{\mathcal{V},3}$, the definition of $K_{2,t}$ implies
\begin{align*}
\| K_{2,t} \text{sh}_{\eta_t} \|_2 \leq& \| \varphi_t \|_4^{1/2} \| \text{sh}_{\eta_t} \|_2 \leq C e^{C |t|},\notag \\
\| \pk_{t} K_{2,t} \text{sh}_{\eta_t} \| \leq& \| \text{sh}_{\eta_t} \|_2 \| \pk_{t} \|_2 \| \varphi_t \|_4^{1/2} \leq C e^{C |t|}
\end{align*}
again from \eqref{eq:regularity phi} and  \eqref{eq:estimates_k,sh,p,r,mu} . The forth term $\mathcal{G}_{2,t}^{\mathcal{V},4}$ satisfies the assumption \eqref{eq:G2_allg} due to \eqref{eq:regularity phi}. 

Furthermore, a transformation of variables shows 
\begin{align*}
\int dxdy \; &\chi (|x-y|\leq \ell ) \; |\varphi_t ((x+y)/2))|^4 \\
&= \int dxdy \; \chi (|x|\leq \ell ) \; |\varphi_t (y)|^4 = C \| \varphi_t \|_4^4 \leq C e^{C |t|}
\end{align*}
Therefore, $\mathcal{G}^{\lambda}_{2,t}$ is of form \eqref{eq:G2_allg}. 

Moreover, for the term $\mathcal{G}_{2,t}^{\mathcal{K}} - \mathcal{K}$ we observe with \eqref{eq:estimates_k,sh,p,r,mu} and \eqref{eq:estimates_delta_eta}
\begin{align*}
\| \Delta_1 p_{t} \|_2 \leq C e^{C |t|}, \; \| \Delta_1 \mu_{t} \|_2 \leq C e^{C |t|}, \;  \| \Delta_1 \rk_{t} \|_2  \leq C e^{C |t|}. 
\end{align*}
The remaining bounds follow in the same way. Note that \eqref{eq:estimates_nabla_eta} implies the bound
\begin{align*}
\| \nabla_1 k_{t} \nabla_1 k_{t}\| \leq \max \left\lbrace  \sup_{x} \int dz \; \vert \nabla_1 k (x;z) \vert , \; \sup_{y} \int dz \; \vert \nabla_1 k (z;y) \vert \right\rbrace \leq C .
\end{align*}
Moreover, by definition \eqref{eq:def_omega_infty} of the limiting kernel, $\omega_\infty \in L^p ( \mathbb{R}^3)$ for all $p<3$. Hence, the remaining terms of $\mathcal{G}_{2,t}^{\mathcal{K}}$ satisfy the assumptions, too.

We are left with the first term of the r.h.s. of \eqref{eq:bogo_G2_1}. We write $T_{t} = e^{ -B( \eta_{t} )}$. The properties \eqref{eq:action_bogo} of the Bogoliubov transformation lead to
\begin{align*}
\left( \partial_t T_{t} \right) T_{t}^* =& - \int_0^1 ds \; e^{-sB( \eta_{t} )} \left( \partial_t B( \eta_{t} )\right) e^{sB( \eta_{t})} \\
=&  \int_0^1 ds \; \int dxdy \;  e^{-s B( \eta_{t})}\left(  \dot{\eta}_{t}(x;y)  a_x^* a_y^* + h.c. \right) e^{s B( \eta_{t}) } \\
=&   \; \int dxdy \;  \dot{\eta}_{t}(x;y) \left(  a^*( \text{ch}_x) a^*( \text{ch}_y) + a( \text{sh}_x) a( \text{sh}_y) \right) + h.c. \\
& + \; \int dxdy \;\dot{\eta}_{t}(x;y)\left(   a^*( \text{ch}_x) a(\text{sh}_y) +  a^*( \text{ch}_y)a( \text{sh}_x)  \right) + h.c.  \\
&+ \int dxdy \; \dot{\eta}_t (x;y) \; \text{sh}_x \text{ch}_y .
\end{align*}
Since $\| \dot{\eta}_t \|_2 \leq C e^{C|t|}$ from \eqref{eq:estimates_nabla_eta}, these terms satisfy assumption \eqref{ass1}, too.  \\

\end{proof}

As proven in \cite[Proposition 8]{BNNS}, any moments of the number of particles operator are approximately preserved with respect to conjugation with the Bogoliubov transformation $T_{N,t}$. To be more precise for every fixed $k \in \mathbb{N}$ and $\delta >0$, there exists $C>0$ such that
\begin{align}
\label{eq:TNT*}
\pm \left( T_{N,t}\mathcal{N}^k T_{N,t}^* - \mathcal{N}^k\right) \leq \delta \mathcal{N}^k + C.
\end{align}


As the following Lemma shows, the moments of number of particles operator are propagated in time with respect to the quadratic $\mathcal{U}_{\mathrm{gen}}( t;0)$. 

\begin{lemma}
\label{lemma:U2NU2^*}
Let $\mathcal{U}_{\mathrm{gen}}(t;s)$ be as defined in Definition \ref{def:U_gen} and $\psi \in \mathcal{F}$. For every $k \in \mathbb{N}$, there exists a constant $C>0$ such that for all $t\in \mathbb{R}$
\begin{align*}
\langle \psi, \mathcal{U}_{\mathrm{gen}}(t;s)^* ( \mathcal{N}+1)^k \; \mathcal{U}_{\mathrm{gen}}(t;s) \psi \rangle \leq C \exp ( C \exp (C |t-s| )) \; \langle \psi, \left( \mathcal{N} + 1\right)^k \psi \rangle . 
\end{align*}
\end{lemma}

\begin{proof} 
We compute the derivative
\begin{align*}
i \frac{d}{dt} &\langle \psi, \mathcal{U}^*_{\mathrm{gen}}(t;s) (\mathcal{N} +1)^k\mathcal{U}_{\mathrm{gen}}(t;s) \psi\rangle \\
=& \langle \psi, \mathcal{U}^*_{\mathrm{gen}}(t;s) \left[ \mathcal{G}_{\mathrm{gen},t} , (\mathcal{N}+1)^k \right] \mathcal{U}_{\mathrm{gen}}(t;s) \psi \rangle  \\
=& \sum_{i=1}^k \langle \psi, \mathcal{U}^*_{\mathrm{gen}}(t;s) (\mathcal{N}+1)^{i-1} \left[ \mathcal{G}_{\mathrm{gen},t} , \mathcal{N} \right] (\mathcal{N}+1)^{k-i}\mathcal{U}_{\mathrm{gen}}(t;s) \psi \rangle.
\end{align*}
Using the commutation relations and the definition \eqref{eq:G2_allg}, we find
\begin{align}
i& \frac{d}{dt} \langle \psi, \mathcal{U}^*_{\mathrm{gen}}(t;s) (\mathcal{N} +1)^k\mathcal{U}_{\mathrm{gen}}(t;s) \psi\rangle\notag \\
=& 2 \sum_{i=1}^k \int dxdy \; H_t^{(2)} (x,y) \; \langle \psi, \; \mathcal{U}^*_{\mathrm{gen}}(t;s)(\mathcal{N}+1)^{i-1} a_x^*a_y^* (\mathcal{N}+1)^{k-i}\mathcal{U}_{\mathrm{gen}}(t;s)\psi \rangle\notag \\
&+ 2 \sum_{i=1}^k \int dxdy \; \overline{H_t^{(2)} (x,y)} \; \langle \psi, \; \mathcal{U}^*_{\mathrm{gen}}(t;s)(\mathcal{N}+1)^{i-1} a_xa_y (\mathcal{N}+1)^{k-i}\mathcal{U}_{\mathrm{gen}}(t;s)\psi \rangle . \label{eq:commL2}
\end{align}
For the first term of the right hand side, the commutation relations yield
\begin{align*}
&\left\vert \int dxdy H_t^{(2)} (x;y) \langle \psi, \mathcal{U}^*_{\mathrm{gen}}(t;s)( \mathcal{N}+1)^{i-1} a_x^* a_y^* (\mathcal{N}+1)^{k-i} \mathcal{U}_{\mathrm{gen}}(t;s) \psi \rangle \right\vert \\
\leq&  \int dxdy | H_t^{(2)} (x;y) | \; \|  ( \mathcal{N}+1)^{(k+1)/2 -i}  a_x^* ( \mathcal{N}+1)^{i-1} \mathcal{U}_{\mathrm{gen}}(t;s) \psi \| \\
& \hspace{3cm} \times \| ( \mathcal{N} +1)^{i-(k-1)/2}a_y( \mathcal{N}+1)^{k-i} \mathcal{U}_{\mathrm{gen}}(t;s) \psi \| \\
\leq & C \| H_t^{(2)} \|_2 \; \| ( \mathcal{N}+1)^{k/2} \mathcal{U}_{\mathrm{gen}}(t;s)\psi \|^2 \leq C e^{C|t|} \;  \langle \psi,\mathcal{U}^*_{\mathrm{gen}}(t;s) ( \mathcal{N}+1)^{k} \mathcal{U}_{\mathrm{gen}}(t;s)\psi \rangle
\end{align*}
where $C$ depends on $k \in \mathbb{N}$. The second of the r.h.s. of \eqref{eq:commL2} follows in the same way. Hence, there exists $C>0$ such that 
\begin{align*}
\left \vert \frac{d}{dt} \langle \psi, \mathcal{U}^*_{\mathrm{gen}}(t;s) (\mathcal{N} +1)^k\mathcal{U}_{\mathrm{gen}}(t;s) \psi\rangle \right\vert \leq C  e^{C|t|}  \langle \psi, \mathcal{U}^*_{\mathrm{gen}}(t;s) (\mathcal{N} +1)^k\mathcal{U}_{\mathrm{gen}}(t;s) \psi\rangle . 
\end{align*}
Hence, the Gronwall inequality implies
\begin{align*}
\langle \psi, \mathcal{U}^*_{\mathrm{gen}}(t;s) (\mathcal{N}  +1)^k\mathcal{U}_{\mathrm{gen}}(t;s) \psi\rangle \leq C \exp \left( C \exp \left( C|t-s|\right) \right) \langle \psi, (\mathcal{N} +1)^k \psi \rangle .
\end{align*}

\end{proof}

For $f \in L^2( \mathbb{R}^3)$ let $\phi_a(f) = a^*(f) + a(f)$. In \cite[Proposition 3.4]{BSS} it is shown that for every $k \in \mathbb{N}$ and $\delta \in \mathbb{R}$ there exists a constant $C>0$ such that 
\begin{align}
\label{eq:phiaN}
\langle \psi, e^{-is \left( \phi_a (f) + \delta  \rd \Gamma (H) \right) } ( \mathcal{N} + 1)^k e^{is\left( \phi_a (f)+ \delta \rd \Gamma (H)\right)} \psi \rangle \leq C \langle \psi, ( \mathcal{N} + \alpha + s^2 \| f \|^2)^k \psi \rangle 
\end{align} 
for all $\psi \in \mathcal{F}$ and $\alpha \geq 1$. Hereafter we denote $\rd \Gamma (H) = \sum_{j=1}^N H^{(j)}$ for a bounded operator $H$ on $L^2( \mathbb{R}^3)$.
A similar estimate holds true for when replacing the creation and annihilation operators $a(f), a^*(f)$ with the modified ones $b^*(f), b(f)$ defined in \eqref{eq:modified c/a}. Let $\phi_b (f) = b^*(f) + b(f)$. In fact, as proven in \cite[Lemma 3.2]{RS_CLT}, for every $k \in \mathbb{N}$ there exists a constant $C>0$ such that
\begin{align}
\label{eq:phibN}
\langle \xi, e^{-i \phi_b(h) } \left( \mathcal{N}_+ (t) + 1\right)^k e^{i \phi_b(h)} \xi \rangle \leq C \langle \xi, \left( \mathcal{N}_+ (t)+ \alpha + \| f \|^2 \right)^k \xi \rangle 
\end{align}
for all $\xi \in \mathcal{F}^{\leq N}_+(t)$ and $\alpha \geq 1$ . 


\subsection{Proof of Proposition \ref{prop:bogo_U2}}
\label{subsec:proof_prop1}

It follows from Lemma \ref{lemma:U2} that it is enough to prove Proposition \ref{prop:bogo_U2} with respect the dynamics $\mathcal{U}_{\mathrm{gen}} (t;s)$.

First, we prove that for $f \in L^2( \mathbb{R}^3)$  the Fock space vectors  $\mathcal{U}^*_{\mathrm{gen}}(t;s) a^*(f) \mathcal{U}_{\mathrm{gen}}(t;s) \Omega $ and $\mathcal{U}_{\mathrm{gen}}^*(t;s) a(f) \mathcal{U}_{\mathrm{gen}}(t;s) \Omega $ are elements of the one-particle sector. The following Lemma is a generalization of \cite[Lemma 8.1]{CLS}.

\begin{lemma}
\label{lemma:UPk}
Let $\mathcal{U}_{\mathrm{gen}} (t;s)$ be the dynamics defined Definition \ref{def:U_gen}. Then for all $f \in L^2( \mathbb{R} )$
\begin{align*}
\mathcal{U}^*_{\mathrm{gen}}(t;s) a^\sharp (f) \mathcal{U}_{\mathrm{gen}}(t;s) \Omega = \mathcal{P}_1 \mathcal{U}^*_{\mathrm{gen}}(t;s) a^\sharp (f) \mathcal{U}_{\mathrm{gen}}(t;s) \Omega
\end{align*}
where either $a^\sharp (f) = a(f)$ or $a^\sharp (f) = a^*(f)$ and where $\mathcal{P}_1$ denotes the projection onto the one particle sector of the Fock space $\mathcal{F}$. 
\end{lemma}

\begin{proof}
The proof follows the arguments of the proof of \cite[Lemma 8.1]{CLS}. For $m \in \mathbb{N}$, $m \not= 1$, we define for arbitrary$m$-particle wave-function $\psi \in \mathcal{F}$ with $\| \psi \| =1$ the function
\begin{align*}
F(t) =& \sup_{\| f \|_2 \leq 1} \vert \langle \psi, \mathcal{U}_{\mathrm{gen}}^*(t;s) a(f) \mathcal{U}_{\mathrm{gen}}(t;s) \Omega \rangle \vert \\
&+ \sup_{\| f \|_2 \leq 1} \vert \langle \psi, \mathcal{U}_{\mathrm{gen}}^*(t;s) a^*(f) \mathcal{U}_{\mathrm{gen}}(t;s) \Omega \rangle \vert .
\end{align*}
Since $m \not= 1$, we observe that $F(s)=0$ and furthermore 
\begin{align*}
e^{i\mathcal{K} t} a(f) e^{-i\mathcal{K}t} =a ( e^{-i\Delta t} f)= a( f_t),
\end{align*}
using the notation $f_t = e^{-it \Delta } f$. 
Since $e^{-i \Delta t}$ is a unitary operator, we find
\begin{align*}
F(t) =&\sup_{\| f \|_2 \leq 1}  \vert \langle \psi, \mathcal{U}_{\mathrm{gen}}^*(t;s) e^{i\mathcal{K} t}a(f)e^{-i\mathcal{K} t} \mathcal{U}_{\mathrm{gen}}(t;s) \Omega \rangle \vert \\
&+ \sup_{ \| f \|_2\leq1}  \vert \langle \psi, \mathcal{U}_{\mathrm{gen}}^*(t;s) e^{i\mathcal{K} t}a^*(f) e^{-i\mathcal{K} t}\mathcal{U}_{\mathrm{gen}}(t;s) \Omega \rangle \vert .
\end{align*}
Then,
\begin{align*}
i \frac{d}{dt}&\langle \psi, \mathcal{U}^*_{\mathrm{gen}}(t;s) e^{i\mathcal{K} t}a(f)e^{-i\mathcal{K} t} \mathcal{U}_{\mathrm{gen}}(t;s) \Omega \rangle \\
 =& \langle \psi, \mathcal{U}_{\mathrm{gen}}^*(t;s) \left[ a(f_t) , \mathcal{G}_{\mathrm{gen},t} - \mathcal{K}\right] \mathcal{U}_{\mathrm{gen}}(t;s) \Omega \rangle
\end{align*}
 and the definition of $\mathcal{G}_{\mathrm{gen},t}$ in \eqref{eq:G2_allg} leads to
\begin{align*}
\left[ a(f_t),  \mathcal{G}_{\mathrm{gen},t}- \mathcal{K} \right] =& \int dxdy \; \left( f_t(x) H_t^{(1)}(x;y)\right) a_y \\
&+  \int dxdy \left( H_t^{(2)}(x;y) f_t(x) + H_t^{(2)}(y;x) f_t (x)\right) a^*_y  .
\end{align*}
The assumption \eqref{ass1} implies on one hand
\begin{align*}
\| H_t^{(1)} f_t \|_2  \leq C e^{C|t|} \; \| f_t \|_2
\end{align*}
and on the other hand
\begin{align*}
\| H_t^{(2)} f_t \|_2 \leq  \| f_t \|_2 \| H_t^{(2)} \|_2 \leq C e^{C|t|}\;  \| f_t \|_2 .
\end{align*}
 Hence,
\begin{align*}
\left\vert\langle \psi, \mathcal{U}^*_{\mathrm{gen}}(t;s) e^{i\mathcal{K} t}a(f)e^{-i\mathcal{K} t} \mathcal{U}_{\mathrm{gen}}(t;s) \Omega \rangle\right\vert \leq C \int_0^t d\tau \; e^{C|\tau|} \; F(\tau)
\end{align*}
and analogously
\begin{align*}
\left\vert \langle \psi, \mathcal{U}^*_{\mathrm{gen}}(t;s) e^{i\mathcal{K} t}a^*(f)e^{-i\mathcal{K} t} \mathcal{U}_{\mathrm{gen}}(t;s) \Omega \rangle\right\vert \leq C\int_0^t d\tau\; e^{C| \tau|} \; F(\tau).
\end{align*}
Note that these bounds are independent of $f \in L^2( \mathbb{R}^3)$. Thus, 
\begin{align*}
0 \leq F(t) \leq C \int_0^t d\tau \;e^{C| \tau |} \; F(\tau).
\end{align*}
Using the bounds $\| a^\sharp(f) \psi \| \leq \| f \|_2 \| ( \mathcal{N}+1)^{1/2} \psi \|$, we obtain
\begin{align*}
F(t) \leq 2 \| ( \mathcal{N}+1)^{1/2} \mathcal{U}_{\mathrm{gen}}(t;s) \Omega \| \leq C \exp \left( \exp ( C |t-s| )\right)  \langle \psi, (\mathcal{N} +1) \psi \rangle .
\end{align*}
Here, we used Lemma \ref{lemma:U2NU2^*}  for the last estimate. Since $F(s)=0$, the Gronwall inequality implies $F(t) =0$ for all $t \in \mathbb{R}$. 
\end{proof}


\begin{proof}[Proof of Proposition \ref{prop:bogo_U2}]. We prove the Proposition with respect to the dynamics $\mathcal{U}_{\mathrm{gen}}(t;s)$ defined inDefinition \ref{def:U_gen}. Then, Proposition \ref{prop:bogo_U2} follows from Lemma \ref{lemma:U2}. \\

The proof follows the arguments of the proof of \cite[Theorem 2.2]{BKS}. Let $\mathcal{P}_k$ denote the projection onto the $k$-particle sector $\mathcal{F}_k$ of the Fock space.  It follows from Lemma that \ref{lemma:UPk}
\begin{align*}
\mathcal{P}_k \mathcal{U}^*_{\mathrm{gen}}(t;s) a^*(f) \mathcal{U}_{\mathrm{gen}}(t;s) \Omega =0, \hspace{0.3cm}
\mathcal{P}_k \mathcal{U}^*_{\mathrm{gen}}(t;s) a^*(f) \mathcal{U}_{\mathrm{gen}}(t;s) \Omega =0
\end{align*}
for all  $f \in L^2( \mathbb{R}^3)$ and $k \not= 1$ . Thus, there exist linear operators $U(t;s), V(t;s) : L^2( \mathbb{R}^3) \rightarrow L^2( \mathbb{R}^3)$ such that
\begin{align*}
\mathcal{U}_{\mathrm{gen}}^* (t;s) a^*(f) \mathcal{U}_{\mathrm{gen}}(t;s) \Omega &= a^* \left( U(t;s) f\right) \Omega, \\
\mathcal{U}_{\mathrm{gen}}^* (t;s) a^*(f) \mathcal{U}_{\mathrm{gen}}(t;s) \Omega &= a^* \left( JV(t;s) f\right) \Omega
\end{align*}
where $J:L^2( \mathbb{R}^3) \rightarrow L^2( \mathbb{R}^3) $ denotes the anti linear operator defined by $Jf= \overline{f}$ for all $f \in L^2( \mathbb{R}^3)$. The operators $U(t;s)$ and $V(t;s)$ are bounded in $L^2( \mathbb{R}^3)$. This follows from Lemma \ref{lemma:U2NU2^*}, since
\begin{align*}
\| U(t;s) f \|&= \| a^*\left( U(t;s)f \right) \Omega \| = \| a^*(f) \mathcal{U}_{\mathrm{gen}}(t;s) \Omega\|\\ 
&\leq \| f \| \; \| ( \mathcal{N}+1)^{1/2} \mathcal{U}_{\mathrm{gen}} (t;s) \Omega \| \leq C \exp (c |t|)
\end{align*}
and
\begin{align*}
\| V(t;s) f \| =& \| a^* \left( JV(t;s) f \right) \Omega \| = \| a(f) \mathcal{U}_{\mathrm{gen}}(t;s) \Omega \| \\
&\leq \| f \| \; \| \mathcal{N}^{1/2} \mathcal{U}_{\mathrm{gen}}(t;s) \Omega \| \leq C e^{K |t| }. 
 \end{align*}
We define the bounded operator $\Theta$ on $ L^2( \mathbb{R}^3) \oplus L^2( \mathbb{R}^3)$ through
\begin{align*}
\Theta (t;s) = \begin{pmatrix}
U(t;s) & JV(t;s) J \\
V(t;s) & JU(t;s) J 
\end{pmatrix}.
\end{align*}
Then 
\begin{align}
\label{eq:bogo_Omega}
\mathcal{U}^*_{\mathrm{gen}} (t;s) A(f,g) \mathcal{U}_{\mathrm{gen}}(t;s) \Omega= A\left( \Theta(t;s) (f,g) \right) \Omega
\end{align}
for all $f,g \in L^2( \mathbb{R}^3)$. For fixed  $\psi \in \mathcal{D}( \mathcal{K} + \mathcal{N} )$, $g \in L^2( \mathbb{R}^3 )$, $s \in \mathbb{R}$ and any bounded operator $\mathcal{M}$ on $\mathcal{F}$ with $\mathcal{M}\mathcal{D} ( \mathcal{K} + \mathcal{N} ) \subset \mathcal{D} ( \mathcal{K} + \mathcal{N} )$, we define furthermore
\begin{align*}
F(t) = \sum_{\sharp} \sup_{\| f \|_2 \leq 1} \left\| \left[ \left[ \mathcal{U}^*_{\mathrm{gen}}(t;s) a^\sharp (f) \mathcal{U}_{\mathrm{gen}} (t;s), \; a^\flat (h) \right], \; \mathcal{M} \right] \; \psi \right\|.
\end{align*}
Here, $a^\sharp, a^\flat$ are either creation or annihilation operators. Since $e^{-i \mathcal{K} t} a^\sharp (f) e^{i \mathcal{K} t} = a^\sharp ( e^{it \Delta} f)$ and $\| e^{it \Delta} f \|_2 = \| f\|_2$ for all$ f \in L^2( \mathbb{R}^3)$, we can write
\begin{align*}
F(t) = \sum_{\sharp} \sup_{\| f \|_2 \leq 1} \left\| \left[ \left[ \mathcal{U}^*_{\mathrm{gen}}(t;s) e^{-i \mathcal{K} (t-s)}a^\sharp(f) e^{i \mathcal{K} (t-s)} \mathcal{U}_{\mathrm{gen}}(t;s), \; a^\flat (h) \right], \; \mathcal{M}  \right] \; \psi  \right\|.
\end{align*}
The commutation relations imply that $F(s)=0$. Furthermore, 
\begin{align*}
i &\frac{d}{dt} \left[ \left[ \mathcal{U}_{\mathrm{gen}}^*(t;s) e^{-i \mathcal{K} (t-s)}a^\sharp(f) e^{i \mathcal{K} (t-s)} \mathcal{U}_{\mathrm{gen}}(t;s), \; a^\flat (h) \right], \; \mathcal{M} \right] \; \psi \\
=& \left[ \left[ \mathcal{U}^*_{\mathrm{gen}}(t;s) \left[ ( \mathcal{G}_{\mathrm{gen},t} - \mathcal{K}), e^{-i \mathcal{K} (t-s)}a^\sharp(f) e^{i \mathcal{K} (t-s)} \right]  \mathcal{U}_{\mathrm{gen}}(t;s), \; a^\flat (h) \right], \; \mathcal{M}  \right] \; \psi\\
=& \left[ \left[ \mathcal{U}^*_{\mathrm{gen}}(t;s) \left[ ( \mathcal{G}_{\mathrm{gen},t} - \mathcal{K}), a^\sharp( e^{-i \Delta (t-s)}f)  \right]  \mathcal{U}_{\mathrm{gen}}(t;s), \; a^\flat (h) \right], \;\mathcal{M}  \right] \; \psi
\end{align*}
using the notation $f_t = e^{-i \Delta t} f$. Analogous calculations as in the proof of Lemma \ref{lemma:UPk} show that
\begin{align*}
\left[ ( \mathcal{G}_{\mathrm{gen},t} - \mathcal{K} ) , a^\sharp (f_t) \right] = a(h_{1,t}) + a^*(h_{2,t}).
\end{align*}
The assumption \eqref{ass1} implies $\|h_{i,t} \|_2 \leq C e^{C |t|}  \| f \|_2 $ for $i =1,2$. Thus,
\begin{align*}
\left\| \left[ \left[ \mathcal{U}^*_{\mathrm{gen}}(t;s) a^\sharp (f) \mathcal{U}_{\mathrm{gen}} (t;s), \; a^\flat (g) \right], \; \mathcal{M} \right] \; \psi \right\| \leq C \int_0^t d \tau \; e^{C | \tau |} \; F( \tau)
\end{align*} 
for all $f \in L^2( \mathbb{R}^3)$ and therefore
\begin{align*}
0 \leq F(t) \leq C \int_0^t d\tau\; e^{C | \tau |} \; F( \tau) .
\end{align*}
Since $F(s)=0$, the Gronwall inequality implies $F(t) = 0 $ for all $t \in \mathbb{R}$. Hence, 
 \begin{align}
 \label{eq:bogo_commM}
\left[ \left[ \mathcal{U}_{\mathrm{gen}}^*(t;s) A(f_1,h_1) \mathcal{U}_{\mathrm{gen}}(t;s), \;A(f_2,h_2) \right], \; \mathcal{M} \right] =0
 \end{align}
 for every $f_1,f_2,h_1,h_2 \in L^2( \mathbb{R}^3)$ and every bounded operator $\mathcal{M}$ on the Fock space $\mathcal{F}$ such that $\mathcal{M} \mathcal{D} \left( \mathcal{K} + \mathcal{N}\right) \subset  \mathcal{D} \left( \mathcal{K} + \mathcal{N} \right)$. We claim, that
 \begin{align}
 \label{eq:bogo_claim1}
 \langle \psi,& \left[ \mathcal{U}^*_{\mathrm{gen}}(t;s) A(f_1,h_1) \mathcal{U}_{\mathrm{gen}}(t;s), \; A(f_2, h_2) \right] \psi \rangle\notag\\
 & = \langle \Omega, \; \left[ \mathcal{U}_{\mathrm{gen}}^*(t;s) A(f_1,h_1) \mathcal{U}_{\mathrm{gen}}(t;s), \; A(f_2, h_2) \right] \Omega \rangle
 \end{align}
for all $\psi \in \mathcal{D}\left( \mathcal{K} + \mathcal{N} \right)$ with $\| \psi \|=1$. Combining \eqref{eq:bogo_Omega} with \eqref{eq:bogo_claim1}, we find
 \begin{align*}
  \langle \psi,& \left[ \mathcal{U}_{\mathrm{gen}}(t;s) A(f_1,h_1) \mathcal{U}_{\mathrm{gen}}(t;s), \; A(f_2, h_2) \right] \psi \rangle \\
  =& \langle \Omega, \; \left[ A\left( \Theta(t;s) (f_1,h_1) \right), \; A(f_2, h_2) \right] \Omega \rangle \\
=& \left( \Theta(t;s) (f_1, h_1), S(f_2, h_2) \right)_{L^2 \oplus L^2} 
 \end{align*}
 where $S$ is defined in \ref{eq:property_A2} . It follows that
 \begin{align}
 \label{eq:bogo_commAA}
\left[ \mathcal{U}^*_{\mathrm{gen}}(t;s) A(f_1,h_1) \mathcal{U}_{\mathrm{gen}}(t;s) - A \left( \Theta(f_1,h_1) \right), \; A(f_2, h_2) \right] =0,
 \end{align}
 for all $f_1,h_1,f_2,h_2 \in L^2( \mathbb{R}^3)$. Consider now 
 $$R:= \mathcal{U}_{\mathrm{gen}}^*(t;s) A(f_1,h_1) \mathcal{U}_{\mathrm{gen}}(t;s) - A \left( \Theta(f_1,h_1) \right). $$
On the one hand, \eqref{eq:bogo_Omega} shows that $R \Omega =0$ and on the other hand it follows from \eqref{eq:bogo_commAA}, that $R$ commutes with any creation and annihilation operator. Since states of the form $a^*(f_1) \dots a^*(f_n) \Omega$ build a basis of the Fock space $\mathcal{F}$, we conclude
 \begin{align*}
 \mathcal{U}^*_{\mathrm{gen}}(t;s) A(f,h) \mathcal{U}_{\mathrm{gen}}(t;s) = A\left( \Theta(t;s) (f,h) \right) 
 \end{align*}
 for all $f,g \in L^2( \mathbb{R}^3)$.

Now, we are left with proving \eqref{eq:bogo_claim1}. For this, note that \eqref{eq:bogo_commM} implies
 \begin{align*}
& \left[ \left[ \mathcal{U}^*_{\mathrm{gen}}(t;s) A(f_1,h_1) \mathcal{U}_{\mathrm{gen}}(t;s), \;A(f_2,h_2) \right], \; P_\psi \right] \\
 &= \left[ \left[ \mathcal{U}^*_{\mathrm{gen}}(t;s) A(f_1,h_1) \mathcal{U}_{\mathrm{gen}}(t;s), \;A(f_2,h_2) \right], \; P_\Omega \right]= 0
 \end{align*}
 where $P_\psi$ resp. $P_\Omega$ denote the projection on the subspace of $\mathcal{F}$ spanned by $\psi$ resp. $\Omega$. Therefore, on one hand
 \begin{align*}
 \langle \psi,  & \left[ \mathcal{U}^*_{\mathrm{gen}}(t;s) A(f_1,h_1) \mathcal{U}_{\mathrm{gen}}(t;s), \;A(f_2,h_2) \right] \Omega \rangle \\
 =&  \langle \psi,   \left[ \mathcal{U}^*_{\mathrm{gen}}(t;s) A(f_1,h_1) \mathcal{U}_{\mathrm{gen}}(t;s), \;A(f_2,h_2) \right] P_\psi \Omega \rangle \\
 =&  \langle \psi,   \left[ \mathcal{U}^*_{\mathrm{gen}}(t;s) A(f_1,h_1) \mathcal{U}_{\mathrm{gen}}(t;s), \;A(f_2,h_2) \right] \psi \rangle \; \langle \psi, \Omega \rangle
 \end{align*}
 and on the other hand
 \begin{align*}
  \langle \psi, &  \left[ \mathcal{U}^*_{\mathrm{gen}}(t;s) A(f_1,h_1) \mathcal{U}_{\mathrm{gen}}(t;s), \;A(f_2,h_2) \right] \Omega \rangle \\
   =&\langle \psi, \Omega \rangle \;   \langle \Omega,   \left[ \mathcal{U}^*_{\mathrm{gen}}(t;s) A(f_1,h_1) \mathcal{U}_{\mathrm{gen}}(t;s), \;A(f_2,h_2) \right] \Omega \rangle.
 \end{align*}
 Assuming that $\langle \psi, \Omega \rangle \not=0$, the claim \eqref{eq:bogo_claim1} follows. If $\langle \psi, \Omega \rangle =0$, we repeat the same arguments with $\widetilde{\psi} = \frac{1}{\sqrt{2}} ( \psi + \Omega)$. This leads to \eqref{eq:bogo_claim1}.

It remains to prove the properties \eqref{eq:bogo_prop}. Since for all $f,g \in L^2( \mathbb{R}^3)$
\begin{align*}
\left( A( \Theta(t;s) (f,h) ) \right)^* =& \left( \mathcal{U}^*_{\mathrm{gen}}(t;s) A( f,h) \mathcal{U}_{\mathrm{gen}}(t;s) \right)^* \\
=& \mathcal{U}^*_{\mathrm{gen}}(t;s) A( f,h)^* \mathcal{U}_{\mathrm{gen}}(t;s)\\
=& \mathcal{U}^*_{\mathrm{gen}}(t;s) A( Jf,Jh) \mathcal{U}_{\mathrm{gen}}(t;s) \\
=& A( \Theta(t;s) (Jf, Jh) )
\end{align*} 
the first  property follows. Furthermore, from 
 \begin{align*}
& \left[ A( \Theta(t;s) (f_1, h_1)), A( \Theta(t;s) (f_2,h_2)) \right] \\
& \hspace{0.2cm} = \left[ \mathcal{U}^*_{\mathrm{gen}}(t;s) A(f_1, h_2) \mathcal{U}_{\mathrm{gen}}(t;s), \mathcal{U}^*(t;s) A(f_2,h_2) \mathcal{U}(t;s)\right] \\
& \hspace{0.2cm} = \mathcal{U}^*_{\mathrm{gen}}(t;s) \left[ A(f_1, h_1), A(f_2,h_2) \right] \mathcal{U}_{\mathrm{gen}}(t;s) \\
 & \hspace{0.2cm}= \langle (f_1,h_1), S(f_2,h_2) \rangle
 \end{align*}
 we deduce the second property.
 
\end{proof}


\subsection{Proof of Theorem \ref{thm:CLT}}

The proof uses ideas introduced in \cite{RS_CLT}. We consider the expectation value
\begin{align*}
\mathbb{E}_{\psit} &\left[ g_1( \mathcal{O}_{1,N,t} ) \dots g_k( \mathcal{O}_{k,N,t}) \right] \notag\\
=& \langle \psit, \; g_1 ( \mathcal{O}_{1,N,t} ) \dots g_k ( \mathcal{O}_{k,N,t} ) \psit \rangle \notag\\
=& \int ds_1 \dots ds_k \; \widehat{g}_1( s_1) \dots \widehat{g}_k ( s_k ) \; \langle \psit, \;  e^{is_1  \mathcal{O}_{1,N,t}} \dots  e^{is_k \mathcal{O}_{k,N,t}} \psit \rangle .
\end{align*}
The norm approximation \eqref{eq:normapprox} from \cite{BNNS} implies, that for every $\alpha < \min \lbrace \beta/2, (1-\beta)/2 \rbrace$ there exists $C>0$ such that
\begin{align}
\Bigg\vert  \mathbb{E}_{\psit} & \left[ g_1( \mathcal{O}_{1,N,t} ) \dots g_k( \mathcal{O}_{k,N,t}) \right] \notag \\
&- \int ds_1 \dots ds_k \;  \widehat{g}_1( s_1) \dots \widehat{g}_k ( s_k ) \notag\\
& \hspace{2cm} \times \langle \mathcal{U}_{\varphi_{N,t}}^* T_{N,t}^* \mathcal{U}_{2}(t;0) \Omega , \; e^{is_1  \mathcal{O}_{1,N,t}} \dots  e^{is_k \mathcal{O}_{k,N,t}} \mathcal{U}_{\varphi_{N,t}}^* T_{N,t}^* \mathcal{U}_{2}(t;0) \Omega \rangle\Bigg\vert \notag\\
\leq& CN^{-\gamma} \prod_{j=1}^k \| \widehat{g}_j \|_1. \label{eq:proof step0_1}
\end{align}
We are hence left with computing the expectation value 
$$\langle \mathcal{U}_{\varphi_{N,t}}^* T_{N,t}^* \mathcal{U}_{2}(t;0) \Omega , \; e^{is_1  \mathcal{O}_{1,N,t}} \dots  e^{is_k \mathcal{O}_{k,N,t}} \mathcal{U}_{\varphi_{N,t}}^* T_{N,t}^* \mathcal{U}_{2}(t;0) \Omega \rangle.$$ 
We split this computation in several Lemmata. \\ 


\begin{lemma}[Action of the unitary $\mathcal{U}_{\varphi_{N,t}}$] \label{lemma:step001} Let $T_{N,t}$ and $\mathcal{U}_2(t;0)$ be as defined in \eqref{eq:def_bogo} resp. \eqref{eq:def_U2}. Moreover, let $\xi_{N,t} = T_{N,t}^* \mathcal{U}_{2}(t;0) \Omega$. Then, using the same notations as in Theorem \ref{thm:CLT}, there exists $C>0$ such that
\begin{align*}
\bigg\vert \langle \mathcal{U}_{\varphi_{N,t}}&  \xi_{N,t}, \;  e^{is_1  \mathcal{O}_{1,N,t}} \dots  e^{is_k \mathcal{O}_{k,N,t}} \mathcal{U}_{\varphi_{N,t}} \xi_{N,t}\rangle\notag \\
&- \langle  \xi_{N,t},\;  e^{is_1  \phi_b\left( q_{N,t} O_{1} \varphi_{N,t}\right) } \dots  e^{is_k \phi_b \left( q_{N,t} O_{k} \varphi_{N,t}\right)}   \xi_{N,t}\rangle\bigg\vert \notag\\
\leq& \frac{C}{\sqrt{N}}\sum_{m=1}^k |s_m|  \| O_m \| \left( 1 + \sum_{j=m}^k s_j^2 \|  O_j \|^2  \right).
\end{align*}
\end{lemma} 

\begin{proof}
We recall that for $f \in L^2( \mathbb{R}^3)$ we denote $\phi_b (f) = b^*(f) + b(f)$ with the modified creation and annihilation operators $b^*(f), b(f)$ defined in \eqref{eq:modified c/a}. 

In order to show Lemma \ref{lemma:step001}, we define for  $j \in \lbrace 1, \cdots , k \rbrace$
\begin{align*}
\widetilde{O}_{j,N,t} = O_j - \langle \varphi_{N,t}, O_j \varphi_{N,t} \rangle.
\end{align*}
We observe that 
\begin{align*}
\mathcal{O}_{j,N,t} = \frac{1}{\sqrt{N}}\left[ \rd \Gamma \left( q_ {N,t} \widetilde{O}_{j,N,t} q_{N,t} \right) + \rd \Gamma \left( p_ {N,t} O_{j} q_{N,t} \right)  + \rd \Gamma \left( q_ {N,t} O_{j} p_{N,t} \right) \right]
\end{align*}
where $p_{N,t} = | \varphi_{N,t} \rangle \langle \varphi_{N,t}\vert$ and $q_{N,t} = 1- p_{N,t}$. 
The properties \eqref{eq:propU} of the unitary $\mathcal{U}_{ \varphi_{N,t}}$ imply
\begin{align*}
\mathcal{U}_{\varphi_{N,t}}^* \mathcal{O}_{j,N,t}\mathcal{U}_{ \varphi_{N,t}} = \frac{1}{\sqrt{N}} \rd \Gamma \left( q_ {N,t} \widetilde{O}_{j,N,t} q_{N,t} \right) + \phi_b \left( q_{N,t} O_j \varphi_{N,t}\right).
\end{align*}
Hence,
\begin{align*}
\langle \mathcal{U}_{\varphi_{N,t}}& \xi_{N,t}, e^{is_1  \mathcal{O}_{1,N,t}} \dots  e^{is_k \mathcal{O}_{k,N,t}} \mathcal{U}_{\varphi_{N,t}} \xi_{N,t} \rangle \\
=& \langle \xi_{N,t}, \prod_{j=1}^k e^{is_j \left(\frac{1}{\sqrt{N}} \rd \Gamma \left( q_ {N,t} \widetilde{O}_{j,N,t} q_{N,t} \right) + \phi_b \left( q_{N,t} O_{j} \varphi_{N,t}\right)\right) } \xi_{N,t} \rangle . 
\end{align*}
We compute  
\begin{align*}
\langle \mathcal{U}_{\varphi_{N,t}}&  \xi_{N,t}, \;  e^{is_1  \mathcal{O}_{1,N,t}} \dots  e^{is_k \mathcal{O}_{k,N,t}} \mathcal{U}_{\varphi_{N,t}} \xi_{N,t}\rangle \\
&- \langle  \xi_{N,t},\;  e^{is_1  \phi_b\left( q_{N,t} O_{1} \varphi_{N,t}\right) } \dots  e^{is_k \phi_b \left( q_{N,t} O_{k} \varphi_{N,t}\right)}   \xi_{N,t}\rangle \\
=&  \sum_{m =1}^k \langle \xi_{N,t}, \; \prod_{j=1}^{m-1} e^{is_j \left(\frac{1}{\sqrt{N}} \rd \Gamma \left( q_ {N,t} \widetilde{O}_{j,N,t} q_{N,t} \right) + \phi_b \left( q_{N,t} O_j \varphi_{N,t}\right)\right) } \\
& \hspace{1.5cm} \times  \left( e^{is_m \left(\frac{1}{\sqrt{N}} \rd \Gamma \left( q_ {N,t} \widetilde{O}_{m,N,t} q_{N,t} \right) + \phi_b \left( q_{N,t} O_{m} \varphi_{N,t}\right) \right) } - e^{is_m  \phi_b \left( \qN O_m \vN \right) } \right) \\
& \hspace{3cm} \times \prod_{j=m+1}^k e^{i s_j  \phi_b( \qN  O_j \vN) } \xi_{N,t} \rangle .
\end{align*}
Using the fundamental theorem of calculus, we can write the difference as an integral
\begin{align*}
\langle \mathcal{U}_{\varphi_{N,t}}&  \xi_{N,t}, \;  e^{is_1  \mathcal{O}_{1,N,t}} \dots  e^{is_k \mathcal{O}_{k,N,t}} \mathcal{U}_{\varphi_{N,t}} \xi_{N,t}\rangle \\
&- \langle  \xi_{N,t},\;  e^{is_1  \phi_b\left( q_{N,t} O_{1} \varphi_{N,t}\right) } \dots  e^{is_k \phi_b \left( q_{N,t} O_{k} \varphi_{N,t}\right)}   \xi_{N,t}\rangle \\
=& \frac{1}{\sqrt{N}}\sum_{m=1}^k \int_0^{s_m} d\tau \; \langle \xi_{N,t}, \; \prod_{j=1}^{m-1} e^{is_j \left(\frac{1}{\sqrt{N}} \rd \Gamma \left( q_ {N,t} \widetilde{O}_{j,N,t} q_{N,t} \right) + \phi_b \left( q_{N,t} O_j \varphi_{N,t}\right)\right) } \\
& \hspace{2cm} \times  e^{i \tau \left(\frac{1}{\sqrt{N}} \rd \Gamma \left( q_ {N,t} \widetilde{O}_{m,N,t} q_{N,t} \right) + \phi_b \left( q_{N,t} O_{m} \varphi_{N,t}\right) \right) } \; \rd \Gamma \left( \qN \widetilde{O}_{m,N,t} \qN \right)  \\
& \hspace{3cm} \times e^{i (1-\tau)  \phi_b \left( \qN O_m \vN \right) }   \prod_{j=m+1}^k e^{i s_j  \phi_b( \qN  O_j \vN) } \xi_{N,t} \rangle  .
\end{align*}
The estimate $\| \rd \Gamma ( A) \psi \| \leq\| A \| \; \| \mathcal{N} \psi \| $ leads to 
\begin{align*}
\bigg\vert \langle \mathcal{U}_{\varphi_{N,t}} & \xi_{N,t}, \;  e^{is_1  \mathcal{O}_{1,N,t}} \dots  e^{is_k \mathcal{O}_{k,N,t}} \mathcal{U}_{\varphi_{N,t}} \xi_{N,t}\rangle \\
&- \langle  \xi_{N,t},\;  e^{is_1  \phi_b\left( q_{N,t} O_{1} \varphi_{N,t}\right) } \dots  e^{is_k \phi_b \left( q_{N,t} O_{k} \varphi_{N,t}\right)}   \xi_{N,t}\rangle\bigg\vert \\
\leq& \frac{1}{\sqrt{N}}\sum_{m=1}^k \int_0^{s_m} d\tau \;   \| \qN \widetilde{O}_{m,N,t} \qN \| \\
& \hspace{2cm} \times \| \mathcal{N} e^{i (1-\tau)  \phi_b \left( \qN O_m \vN \right) }   \prod_{j=m+1}^k e^{i s_j  \phi_b( \qN  O_j \vN) } \xi_{N,t}\|
\end{align*}
Since $\| \qN \widetilde{O}_{m,N,t} \qN \| \leq \| O_m \|$ and $\| \qN O_j \vN \| \leq \| O_j \|$, we find with \eqref{eq:phibN} 
\begin{align*}
\bigg\vert \langle& \mathcal{U}_{\varphi_{N,t}}  \xi_{N,t}, \;  e^{is_1  \mathcal{O}_{1,N,t}} \dots  e^{is_k \mathcal{O}_{k,N,t}} \mathcal{U}_{\varphi_{N,t}} \xi_{N,t}\rangle \\
&- \langle  \xi_{N,t},\;  e^{is_1  \phi_b\left( q_{N,t} O_{1} \varphi_{N,t}\right) } \dots  e^{is_k \phi_b \left( q_{N,t} O_{k} \varphi_{N,t}\right)}   \xi_{N,t}\rangle\bigg\vert \\
\leq& \frac{1}{\sqrt{N}}\sum_{m=1}^k \| O_m \| \int_0^{s_m} d\tau \;    \| \left( \mathcal{N}_+ +(1- \tau)^2 \|  O_m  \|^2 + \sum_{j=m+1}^k s_j^2 \|  O_j \|^2 + \alpha \right)  \xi_{N,t}\|
\end{align*}
for $\alpha \geq 1$. Recall that $\xi_{N,t} = T_{N,t}^* \mathcal{U}_2 (t;0) \Omega $. It follows from \eqref{eq:TNT*} and Lemma \ref{lemma:U2NU2^*} that
\begin{align*}
\langle \xi_{N,t}, \mathcal{N}_+^2 \xi_{N,t} \rangle \leq C
\end{align*}
for a constant $C>0$ uniform in $N$. Hence, 
\begin{align*}
\bigg\vert \langle \mathcal{U}_{\varphi_{N,t}}&  \xi_{N,t}, \;  e^{is_1  \mathcal{O}_{1,N,t}} \dots  e^{is_k \mathcal{O}_{k,N,t}} \mathcal{U}_{\varphi_{N,t}} \xi_{N,t}\rangle \\
&- \langle  \xi_{N,t},\;  e^{is_1  \phi_b\left( q_{N,t} O_{1} \varphi_{N,t}\right) } \dots  e^{is_k \phi_b \left( q_{N,t} O_{k} \varphi_{N,t}\right)}   \xi_{N,t}\rangle\bigg\vert \\
\leq& \frac{C}{\sqrt{N}}\sum_{m=1}^k |s_m|  \| O_m \| \left( 1 + \sum_{j=m}^k s_j^2 \|  O_j \|^2  \right).
\end{align*}
\end{proof}


\begin{lemma}[ Replace modified creation and annihilation operators with standard ones] Let $T_{N,t}$ and $\mathcal{U}_2(t;0)$ be as defined in \eqref{eq:def_bogo} resp. \eqref{eq:def_U2}. Moreover, let $\xi_{N,t} = T_{N,t}^* \mathcal{U}_{2}(t;0) \Omega$. Then, with the same notations as in Theorem \ref{thm:CLT}, there exists $C>0$ such that 
\begin{align*}
\bigg\vert \langle  \xi_{N,t} &,\;  e^{is_1  \phi_b\left( q_{N,t} O_{1} \varphi_{N,t}\right) } \dots  e^{is_k \phi_b \left( q_{N,t} O_{k} \varphi_{N,t}\right)}   \xi_{N,t} \rangle \\
&- \langle  \xi_{N,t},\;  e^{is_1  \phi_a\left( q_{N,t} O_{1} \varphi_{N,t}\right) } \dots  e^{is_k \phi_a \left( q_{N,t} O_{k} \varphi_{N,t}\right)}   \xi_{N,t}\rangle \bigg\vert  \\
\leq& \frac{C}{N} \sum_{m=1}^k \| O_m \| |s_m| \left(1 + \sum_{j=m}^ks_j^2  \| O_j \|  \right)^{3/2} .
\end{align*} 

\end{lemma}

\begin{proof}
Recall that 
$$\phi_a (f) = a^*(f) + a(f)$$
with the standard creation and annihilation operators $a^*(f), a(f)$ while 
$$\phi_b (f) = b^*(f) + b(f)$$
with the modified creation and annihilation operators defined in \eqref{eq:modified c/a}. To this end, we compute 
\begin{align*}
\langle  \xi_{N,t}&,\;  e^{is_1  \phi_b\left( q_{N,t} O_{1} \varphi_{N,t}\right) } \dots  e^{is_k \phi_b \left( q_{N,t} O_{k} \varphi_{N,t}\right)}   \xi_{N,t}\rangle \\
&- \langle  \xi_{N,t},\;  e^{is_1  \phi_a\left( q_{N,t} O_{1} \varphi_{N,t}\right) } \dots  e^{is_k \phi_a \left( q_{N,t} O_{k} \varphi_{N,t}\right)}   \xi_{N,t}^{(1)}\rangle  \\
=& \sum_{m=1}^k \langle \xi_{N,t}, \; \prod_{j=1}^{m-1} e^{i s_j \phi_b \left( q_{N,t} O_{j} \varphi_{N,t}\right) } \\
& \hspace{2cm} \times \left( e^{i s_m \phi_b \left( q_{N,t} O_{m} \varphi_{N,t}\right) } - e^{i s_m \phi_a \left( q_{N,t} O_{m} \varphi_{N,t}\right) }   \right) \\
& \hspace{3cm} \times  \prod_{j=m+1}^{k} e^{i s_j \phi_a \left( q_{N,t} O_{j} \varphi_{N,t}\right) } \xi_{N,t} \rangle  \\
=& \sum_{m=1}^k \int_0^{s_m} d\tau \; \langle \xi_{N,t}, \;  \prod_{j=1}^{m-1} e^{i s_j \phi_b \left( q_{N,t} O_{j} \varphi_{N,t}\right) } e^{i \tau \phi_b \left( q_{N,t} O_{m} \varphi_{N,t}\right) } \\
& \hspace{3cm} \times \left( \phi_b \left( q_{N,t} O_{m} \varphi_{N,t}\right) - \phi_a\left(  q_{N,t} O_{m} \varphi_{N,t}\right) \right) \\
& \hspace{4cm} \times e^{i (1-\tau) \phi_a \left( q_{N,t} O_{m} \varphi_{N,t}\right) }\prod_{j=m+1}^{k} e^{i s_j \phi_a \left( q_{N,t} O_{j} \varphi_{N,t}\right) } \xi_{N,t} \rangle .
\end{align*} 
By definition of the modified creation and annihilation operators \eqref{eq:modified c/a} we obtain
\begin{align*}
\langle  \xi_{N,t}&,\;  e^{is_1  \phi_b\left( q_{N,t} O_{1} \varphi_{N,t}\right) } \dots  e^{is_k \phi_b \left( q_{N,t} O_{k} \varphi_{N,t}\right)}   \xi_{N,t}\rangle \\
&- \langle  \xi_{N,t},\;  e^{is_1  \phi_a\left( q_{N,t} O_{1} \varphi_{N,t}\right) } \dots  e^{is_k \phi_a \left( q_{N,t} O_{k} \varphi_{N,t}\right)}   \xi_{N,t}\rangle  \\
=& \sum_{m=1}^k \int_0^{s_m} d\tau \; \langle \xi_{N,t}^{(1)}, \;  \prod_{j=1}^{m-1} e^{i s_j \phi_b \left( q_{N,t} O_{j} \varphi_{N,t}\right) } e^{i \tau \phi_b \left( q_{N,t} O_{m} \varphi_{N,t}\right) } \\
& \hspace{3cm} \times  a^* \left( q_{N,t} O_{m} \varphi_{N,t}\right) \left(\sqrt{1 - \mathcal{N}_+/N}  -1 \right) \\
& \hspace{4cm} \times e^{i (1-\tau) \phi_a \left( q_{N,t} O_{m} \varphi_{N,t}\right) }\prod_{j=m+1}^{k} e^{i s_j \phi_a \left( q_{N,t} O_{j} \varphi_{N,t}\right) } \xi_{N,t} \rangle \\
&+ \sum_{m=1}^k \int_0^{s_m} d\tau \; \langle \xi_{N,t}, \;  \prod_{j=1}^{m-1} e^{i s_j \phi_b \left( q_{N,t} O_{j} \varphi_{N,t}\right) } e^{i \tau \phi_b \left( q_{N,t} O_{m} \varphi_{N,t}\right) } \\
& \hspace{3cm} \times  \left(\sqrt{1 - \mathcal{N}_+/N}  -1 \right) a \left( q_{N,t} O_{m} \varphi_{N,t}\right)  \\
& \hspace{4cm} \times e^{i (1-\tau) \phi_a \left( q_{N,t} O_{m} \varphi_{N,t}\right) }\prod_{j=m+1}^{k} e^{i s_j \phi_a \left( q_{N,t} O_{j} \varphi_{N,t}\right) } \xi_{N,t} \rangle .
\end{align*} 
Since $\| a^*(f) \xi \| \leq \|f \|_2 \; \| ( \mathcal{N} +1)^{1/2} \xi \|$ resp. $\| a(f) \xi \| \leq \|f \|_2 \; \| \mathcal{N}^{1/2} \xi \|$ and $\| \qN O_m \vN \|_2 \leq \| O_m \|$, we find 
\begin{align*}
\bigg\vert \langle  \xi_{N,t}&,\;  e^{is_1  \phi_b\left( q_{N,t} O_{1} \varphi_{N,t}\right) } \dots  e^{is_k \phi_b \left( q_{N,t} O_{k} \varphi_{N,t}\right)}   \xi_{N,t}\rangle \\
&- \langle  \xi_{N,t},\;  e^{is_1  \phi_a\left( q_{N,t} O_{1} \varphi_{N,t}\right) } \dots  e^{is_k \phi_a \left( q_{N,t} O_{k} \varphi_{N,t}\right)}   \xi_{N,t}\rangle \bigg\vert  \\
\leq& \frac{2}{N} \sum_{m=1}^k \| O_m \| \; \int_0^{s_m} d\tau \;   \| ( \mathcal{N} +1)^{3/2} e^{i (1-\tau) \phi_a \left( q_{N,t} O_{m} \varphi_{N,t}\right) }\prod_{j=m+1}^{k} e^{i s_j \phi_a \left( q_{N,t} O_{j} \varphi_{N,t}\right) } \xi_{N,t}  \| .
\end{align*} 
Now, Lemma \ref{lemma:U2NU2^*} together with \eqref{eq:phiaN} and \eqref{eq:TNT*} implies
\begin{align*}
\bigg\vert \langle  \xi_{N,t}&,\;  e^{is_1  \phi_b\left( q_{N,t} O_{1} \varphi_{N,t}\right) } \dots  e^{is_k \phi_b \left( q_{N,t} O_{k} \varphi_{N,t}\right)}   \xi_{N,t}\rangle \\
&- \langle  \xi_{N,t},\;  e^{is_1  \phi_a\left( q_{N,t} O_{1} \varphi_{N,t}\right) } \dots  e^{is_k \phi_a \left( q_{N,t} O_{k} \varphi_{N,t}\right)}   \xi_{N,t}\rangle \bigg\vert  \\
\leq& \frac{C}{N} \sum_{m=1}^k \| O_m \| |s_m| \left(1 + \sum_{j=m}^ks_j^2  \| O_j \|  \right)^{3/2} .
\end{align*} 

\end{proof}


\begin{lemma}[Replace modified Hartree equation with non-linear Schr\"odinger equation] Let $T_{N,t}$ and $\mathcal{U}_2(t;0)$ be as defined in \eqref{eq:def_bogo} resp. \eqref{eq:def_U2}. Moreover, let  $\xi_{N,t} = T_{N,t}^* \mathcal{U}_{2}(t;0) \Omega$. Then, with the same notations as in Theorem \ref{thm:CLT}, there exists $C>0$ such that 
\begin{align*}
\bigg\vert \langle  \xi_{N,t},\; & e^{is_1  \phi_a\left( q_{N,t} O_{1} \varphi_{N,t}\right) } \dots  e^{is_k \phi_a \left( q_{N,t} O_{k} \varphi_{N,t}\right)}   \xi_{N,t}\rangle - \langle  \xi_{N,t}^{(1)},\;  e^{is_1  \phi_a\left( q_t O_{1} \varphi_{t}\right) } \dots  e^{is_k \phi_a \left( q_{t} O_{k} \varphi_{t}\right)}   \xi_{N,t}^{(1)}\rangle\bigg\vert  \\
\leq& CN^{-\gamma} \sum_{m=1}^k |s_m| \; \| O_m \| \left( 1 + \sum_{j=m}^k s_j^2 \| O_j \|^2 \right)^{1/2} \exp \left( \exp \left( C |t| \right) \right).
\end{align*}
\end{lemma}

\begin{proof}
By linearity of the operator $\phi_a (f)$, we compute
\begin{align*}
\langle  \xi_{N,t},\; & e^{is_1  \phi_a\left( q_{N,t} O_{1} \varphi_{N,t}\right) } \dots  e^{is_k \phi_a \left( q_{N,t} O_{k} \varphi_{N,t}\right)}   \xi_{N,t}\rangle - \langle  \xi_{N,t},\;  e^{is_1  \phi_a\left( q_t O_{1} \varphi_{t}\right) } \dots  e^{is_k \phi_a \left( q_{t} O_{k} \varphi_{t}\right)}   \xi_{N,t}\rangle  \\
=& \sum_{m=1}^k \langle  \xi_{N,t},\; \prod_{j=1}^{m-1} e^{is_j  \phi_a\left( q_{N,t} O_{j} \varphi_{N,t}\right) } \left(  e^{is_m  \phi_a\left( q_{N,t} O_{m} \varphi_{N,t}\right) }  -  e^{is_m  \phi_a\left( q_{t} O_{m} \varphi_{t}\right) } \right)  \\
& \hspace{3cm} \times \prod_{j=m+1}^{k} e^{is_j \phi_a\left( q_{N,t} O_{j} \varphi_{N,t}\right) } \xi_{N,t} \rangle \\
=& \sum_{m=1}^k \int_0^{s_m} d\tau \; \langle  \xi_{N,t},\; \prod_{j=1}^{m-1} e^{is_j  \phi_a\left( q_{N,t} O_{j} \varphi_{N,t}\right) }    \\
&\hspace{2.7cm} \times e^{i\tau   \phi_a\left( q_{N,t} O_{m} \varphi_{N,t}\right) }  \phi_a \left( q_{N,t} O_m \varphi_{N,t} - q_t O_m \varphi_t \right)  e^{i (1-\tau) \phi_a\left( q_{t} O_{m} \varphi_{t}\right) } \\
& \hspace{4cm} \times \prod_{j=m+1}^{k} e^{is_j \phi_a\left( q_{N,t} O_{j} \varphi_{N,t}\right) } \xi_{N,t} \rangle \\
\end{align*}
As
\begin{align*}
\| q_{N,t} O_m \varphi_{N,t} - q_t O_m \varphi_t \|_2 \leq \| O_m \| \left( \| q_{N,t} - q_t \|_2 + \| \varphi_{N,t} - \varphi_t \|_2 \right) \leq 2 \| O_m \| \| \varphi_{N,t} - \varphi_t \|
\end{align*}
the estimate \eqref{eq:approx phi} implies 
\begin{align*}
\| q_{N,t} O_m \varphi_{N,t} - q_t O_m \varphi_t \|_2 \leq C \| O_m \| N^{-\gamma} \exp \left( \exp \left( C |t| \right) \right)
\end{align*}
with $\gamma =\min \lbrace \beta, 1-\beta \rbrace$. Hence, the bound 
$$
\| \phi_a (f) \psi \| \leq 2 \| f \|_2 \; \| \left( \mathcal{N} +1\right)^{1/2} \psi \|
$$
leads to
\begin{align*}
\bigg\vert \langle  \xi_{N,t},\; & e^{is_1  \phi_a\left( q_{N,t} O_{1} \varphi_{N,t}\right) } \dots  e^{is_k \phi_a \left( q_{N,t} O_{k} \varphi_{N,t}\right)}   \xi_{N,t}\rangle - \langle  \xi_{N,t},\;  e^{is_1  \phi_a\left( q_t O_{1} \varphi_{t}\right) } \dots  e^{is_k \phi_a \left( q_{t} O_{k} \varphi_{t}\right)}   \xi_{N,t}\rangle\bigg\vert  \\
\leq& CN^{-\gamma} \exp \left( \exp \left( C |t| \right) \right) \sum_{m=1}^k \| O_m \| \int_0^{s_m} d\tau \;  \\
& \hspace{3cm} \times \| \left( \mathcal{N} +1\right)^{1/2} e^{i (1-\tau) \phi_a\left( q_{t} O_{m} \varphi_{t}\right) }   \prod_{j=m+1}^{k} e^{is_j \phi_a\left( q_{N,t} O_{j} \varphi_{N,t}\right) } \xi_{N,t} \|.
\end{align*}
We conclude again with Lemma \eqref{eq:phiaN}, Lemma \eqref{eq:TNT*} and Lemma \ref{lemma:U2NU2^*}
\begin{align*}
\bigg\vert \langle  \xi_{N,t},\; & e^{is_1  \phi_a\left( q_{N,t} O_{1} \varphi_{N,t}\right) } \dots  e^{is_k \phi_a \left( q_{N,t} O_{k} \varphi_{N,t}\right)}   \xi_{N,t}\rangle - \langle  \xi_{N,t}^{(1)},\;  e^{is_1  \phi_a\left( q_t O_{1} \varphi_{t}\right) } \dots  e^{is_k \phi_a \left( q_{t} O_{k} \varphi_{t}\right)}   \xi_{N,t}^{(1)}\rangle\bigg\vert  \\
\leq& CN^{-\gamma} \sum_{m=1}^k |s_m| \; \| O_m \| \left( 1 + \sum_{j=m}^k s_j^2 \| O_j \|^2 \right)^{1/2} \exp \left( \exp \left( C |t| \right) \right).
\end{align*}
\end{proof}


\begin{lemma}[Action of $T_{N,t}$] Let $T_{N,t}$ and $\mathcal{U}_2(t;0)$ be as defined in \eqref{eq:def_bogo} resp. \eqref{eq:def_U2}. Moreover, let  $\xi_{N,t} = T_{N,t}^* \mathcal{U}_{2}(t;0) \Omega$ and $\xi_{t} =T_{N,t} \xi_{N,t}= \mathcal{U}_{2,t} \Omega$. Then, using the same notations as in Theorem \ref{thm:CLT}, there exists $C>0$ such that 
\begin{align*}
\bigg\vert \langle & \xi_{N,t},\;  e^{is_1  \phi_a\left( q_t O_{1} \varphi_{t}\right) } \dots  e^{is_k \phi_a \left( q_{t} O_{k} \varphi_{t}\right)}   \xi_{N,t}\rangle - \langle \xi_t, e^{is_1  \phi_a\left( h_{1,t}\right) } \dots  e^{is_k \phi_a \left( h_{k,t}\right)} \xi_t \rangle \bigg\vert \\
& \leq C N^{-\gamma} \sum_{m=1}^k |s_m| \| O_m \| \left( 1 + \sum_{j=m}^k s_j^2 \| O_j \|^2 \right)^{1/2} \exp \left( \exp \left( C |t \right) \right) .
\end{align*}
with $h_{j,t} = \cosh( \eta_t )  q_t O \varphi_t + \sinh ( \eta_t ) \overline{q_t O_j \varphi_t}$ and $\eta_t$ as defined in \eqref{eq:def eta}.
\end{lemma}

\begin{proof}
We compute using the properties \eqref{eq:action_bogo} of the Bogoliubov transformation
\begin{align*}
T_{N,t}^* \phi_a \left( q_t O_j \varphi_t \right) T_{N,t} = \phi_a \left( \cosh( \eta_{N,t} )  q_t O_j \varphi_t + \sinh ( \eta_{N,t} ) \overline{q_t O_j \varphi_t} \right). 
\end{align*}
with $\eta_{N,t}$ as defined in \eqref{eq:def eta}. In the following we denote $h_{j,N,t} =\cosh( \eta_{N,t} )  q_t O \varphi_t + \sinh ( \eta_{N,t} ) \overline{q_t O_j \varphi_t}  $. Since 
\begin{align*}
 \langle  \xi_{N,t},\;  e^{is_1  \phi_a\left( q_t O_{1} \varphi_{t}\right) } \dots  e^{is_k \phi_a \left( q_{t} O_{k} \varphi_{t}\right)}   \xi_{N,t}\rangle =  \langle  \xi_{t},\;  e^{is_1  \phi_a\left( h_{1,N,t } \right) } \dots  e^{is_k \phi_a \left( h_{k,N,t}\right)}   \xi_{t} \rangle
\end{align*}
we need to consider
\begin{align*}
 \langle  \xi_{t},\;  e^{is_1  \phi_a\left(h_{1,N,t } \right) } \dots  e^{is_k \phi_a \left( h_{k,N,t}\right)}   \xi_{t} \rangle - \langle  \xi_{t},\;  e^{is_1  \phi_a\left( h_{1,t } \right) } \dots  e^{is_k \phi_a \left( h_{k,t}\right)}   \xi_{t} \rangle .
\end{align*}
We observe using \eqref{eq:estimates_nabla_eta}
\begin{align*}
\| h_{j,N,t} - h_{j,t} \|_2 \leq& \| O_m \| \left( \| \cosh( \eta_t) - \cosh (\eta_{N,t} )\|_2 + \| \sinh (\eta_t ) - \sinh ( \eta_{N,t} \|_2  \right) \\
\leq& 2 \| O_m \| \cosh(( \eta_{N,t} + \eta_t)/2) \; \sinh (( \eta_{N,t} - \eta_t)/2) \|_2 \\
&+ 2 \| O_m \| \sinh( (\eta_{N,t} + \eta_t)/2) \; \sinh ( (\eta_{N,t} - \eta_t)/2) \|_2 \\
\leq& C \| O_m \| \| \eta_{N,t} - \eta_t \|_2.
\end{align*}
Thus, the estimate \eqref{eq:approx eta} leads to
\begin{align*}
\| h_{j,N,t} - h_{j,t} \|_2 \leq& C N^{-\gamma } \exp \left( \exp \left( C |t| \right) \right) .
\end{align*}
Using $\| h_{j,t} \|_2 \leq C \| O_j \| $, the same arguments as in step 3 lead to
\begin{align*}
\bigg\vert \langle & \xi_{N,t},\;  e^{is_1  \phi_a\left( q_t O_{1} \varphi_{t}\right) } \dots  e^{is_k \phi_a \left( q_{t} O_{k} \varphi_{t}\right)}   \xi_{N,t}\rangle - \langle \xi_t, e^{is_1  \phi_a\left(h_{1,t}\right) } \dots  e^{is_k \phi_a \left( h_{k,t}\right)} \xi_t \rangle \bigg\vert \\
& \leq C N^{-\gamma} \sum_{m=1}^k |s_m| \| O_m \| \left( 1 + \sum_{j=m}^k s_j^2 \| O_j \|^2 \right)^{1/2} \exp \left( \exp \left( C |t \right) \right) .
\end{align*}
\end{proof}


\begin{lemma}[Computing the expectation value] 
\label{lemma:step005} Let $\mathcal{U}_2(t;0)$ be as defined in \eqref{eq:def_U2}. Let $\xi_t = \mathcal{U}_{2,t} \Omega$. Then, using the same notations as in Theorem \ref{thm:CLT}, 
\begin{align*}
 \langle \xi_t, &e^{is_1  \phi_a\left( h_{1,t}\right) } \dots  e^{is_k \phi_a \left( h_{k,t}\right)} \xi_t \rangle \\
 =& \frac{1}{\sqrt{(2\pi)^k \det \Sigma_t}} \int dx_1 \dots dx_k \; g_1( x_1) \dots g_k (x_k) \; e^{-\frac{1}{2} \sum_{i,j =1}^k \left(\Sigma_t \right)_{i,j}^{-1} x_i x_j}.
 \end{align*}
\end{lemma}

\begin{proof}
 Since $\xi_t = \mathcal{U}_{2,t} \Omega$, we are left with computing 
 \begin{align}
 \langle \xi_t, e^{is_1  \phi_a\left( h_{1,t}\right) } \dots  e^{is_k \phi_a \left( h_{k,t}\right)} \xi_t \rangle =&  \langle \Omega, \mathcal{U}_{2,t}^*  e^{is_1  \phi_a\left(h_{1,t}\right) } \dots  e^{is_k \phi_a \left( h_{k,t}\right)} \mathcal{U}_{2,t} \Omega \rangle .
 \label{eq:expectationvalue}
 \end{align}
As proven in Proposition \ref{prop:bogo_U2}, the unitary $\mathcal{U}_{2,t}$ gives rise to a Bogoliubov transformation. Hence, there exists bounded operators $U(t;0),V(t;0)$ on $L^2( \mathbb{R}^3)$ such that 
\begin{align*}
\mathcal{U}_{2,t}^* \phi_a\left(h_{j,t} \right) \mathcal{U}_{2,t} =\phi_a \left( U(t;0) h_{j,t} + \overline{V(t;0)} \overline{h}_{j,t} \right). 
\end{align*}
In the following we denote 
\begin{align*}
\nu_{j,t} =&U(t;0) h_{j,t} + \overline{V(t;0)} \overline{h}_{j,t} \\
=& \left( U(t;0) \cosh \eta_t + \overline{V(t;0)} \sinh \eta_t \right) q_t O_j \varphi_t \\
&+ \left( U(t;0) \sinh \eta_t + \overline{V(t;0)} \cosh \eta_t \right) \overline{q_t O_j \varphi_t }.
\end{align*}
Note, that the Baker Campbell Hausdorff formula implies on one hand
\begin{align*}
e^{i \phi_a (f)} e^{i \phi_a (g)} = e^{i \phi_a ( f + g) } e^{-i \text{Im} \langle f,g \rangle}
\end{align*}
for $f,g \in L^2( \mathbb{R}^3 )$, i.e.
\begin{align*}
\prod_{j=1}^k e^{i s_j \phi_a ( \nu_{j,t} ) } = e^{i \phi_a ( \nu_t )} \prod_{i<j}^k e^{-is_i s_j \text{Im} \langle \nu_{i,t},\nu_{j,t} \rangle } 
\end{align*}
with $\nu_t= \sum_{j=1}^k \nu_{j,t}$. On the other hand, the Baker Campbell Hausdorff formula applied to the creation and annihilation operator implies
\begin{align*}
\prod_{j=1}^k e^{i s_j \phi_a ( \nu_{j,t} ) } = e^{- \| \nu_t\|_2^2/2} e^{a^*(\nu_t)} e^{a(\nu_t)} \prod_{i<j}^k e^{-is_i s_j \text{Im} \langle \nu_{i,t}, \nu_{j,t} \rangle } .
\end{align*}
Hence, we write the expectation value \eqref{eq:expectationvalue} as
\begin{align*}
\langle \xi_{t}, e^{i \phi_a ( h_{1,t})} \dots e^{i s_k \phi_a (h_{k,t})} \xi_{t} \rangle &= e^{- \| \nu_t \|_2^2}\prod_{i<j}^k e^{-is_i s_j \text{Im} \langle \nu_{i,t}, \nu_{j,t} \rangle } \langle \Omega, e^{a^*(\nu_t) } e^{a(\nu_t)} \Omega \rangle \\
&= e^{- \| \nu_t \|_2^2}\prod_{i<j}^k e^{-is_i s_j \text{Im} \langle \nu_{i,t}, \nu_{j,t} \rangle }.
\end{align*}
Let $\Sigma_t \in \mathbb{C}^{k \times k}$ be given through
\begin{align*}
\left( \Sigma_{t}\right)_{i,j} = \begin{cases}
\langle \nu_{i,t}, \nu_{j,t}\rangle, &\text{if} \hspace{0.2cm} i<j \\
\langle \nu_{j,t}, \nu_{i,t}\rangle,&  \text{otherwise}
\end{cases} 
\end{align*}
then,
\begin{align*}
&\langle \xi_{t}, e^{i s_1 \phi_a ( h_{1,t})} \dots e^{i s_k \phi_a (h_{k,t})} \xi_{t} \rangle = e^{-\frac{1}{2} \sum_{i,j=1}^k \left( \Sigma_t \right) _{i,j} s_is_j}.
\end{align*}
By assumption, the matrix $\Sigma_t $ is invertible. Hence,
\begin{align*}
&\int ds_1 \dots ds_k \; \widehat{g}_1 (s_1) \dots \widehat{g}_k (s_k) \; \langle \xi_{t}, e^{i s_1 \phi_a (h_{1,t})} \dots e^{i s_k \phi_a (h_{k,t})} \xi_{t} \rangle   \\
=&  \int ds_1 \dots ds_k \; \widehat{g}_1 (s_1) \dots \widehat{g}_k (s_k) \;e^{-\frac{1}{2} \sum_{i,j=1}^k \left(\Sigma_t \right)_{i,j} s_i s_j} \\
=& \frac{1}{\sqrt{(2\pi)^k \det \Sigma_t}} \int dx_1 \dots dx_k \; g_1( x_1) \dots g_k (x_k) \; e^{-\frac{1}{2} \sum_{i,j =1}^k \left(\Sigma_t \right)_{i,j}^{-1} x_i x_j} .
\end{align*}
\end{proof}

Summarizing the results from Lemma \ref{lemma:step001} to Lemma \ref{lemma:step005}, we finally obtain
\begin{align*}
\bigg\vert& \mathbb{E}_{\Psi_{N,t}} \left[ g_1( \mathcal{O}_{1,N,t}) \dots g_k ( \mathcal{O}_{k,N,t} ) \right] \\
& \hspace{2cm}- \frac{1}{\sqrt{(2\pi)^k \det \Sigma_t }} \int dx_1 \dots dx_k \; g_1( x_1) \dots g_k (x_k) \; e^{-\frac{1}{2} \sum_{i,j =1}^k \left(\Sigma_t \right)_{i,j}^{-1} x_i x_j} \bigg\vert \\
& \leq C N^{-\gamma} \exp( C \exp (C|t|))  \\
& \hspace{1cm} \times \prod_{j=1}^k \int d\tau | \widehat{g}_j (\tau) \vert \left( 1 + |\tau|^2 \| O_j \|^2  +N^{\gamma-1/2} |\tau|^3 \| O_j \|^3 + N^{\gamma -1} | \tau |^4 \| O_j \|^4 \right) .
\end{align*}
with $\gamma = \min \lbrace \beta, 1- \beta \rbrace$.  This proves Theorem \ref{thm:CLT}. 

\vspace{1cm} 
\textit{\textbf{Acknowledgement.}} S.R. acknowledges partial support from the NCCR SwissMAP. This project has received funding from the European Union's Horizon 2020 research and innovation programme under the Marie Sk\l{}odowska-Curie Grant Agreement No. 754411.  \\
S.R. would like to thank Benjamin Schlein for many fruitful discussions.


\vspace{1cm}


\begin{thebibliography}{}
%
%

\bibitem{AGT}
R.~Adami, F.~Golse, and A.~Teta, Rigorous derivation of the cubic {NLS} in dimension one, J. Stat. Phys., 127(6):1193--1220 (2007)

\bibitem{AFP}
Z.~Ammari, M.~Falconi, and B.~Pawilowski, On the rate of convergence for the mean field approximation of bosonic many-body quantum dynamics, Commun. Math. Sci., 14(5):1417--1442 (2016).

\bibitem{AN}
Z.~Ammari and F.~Nier.
\newblock Mean field limit for bosons and propagation of {W}igner measures, J. Math. Phys., 50(4):042107, 16 (2009)

\bibitem{AH}
I.~Anapolitanos and M.~Hott, A simple proof of convergence to the {H}artree dynamics in {S}obolev trace norms, J. Math. Phys., 57(12):122108, 8 (2016)

\bibitem{BGM}
C.~Bardos, F.~Golse, and N.~J. Mauser, Weak coupling limit of the {$N$}-particle {S}chr\"{o}dinger equation, Methods Appl. Anal., 7(2):275--293 (2000)

\bibitem{BKS}
G.~Ben~Arous, K.~Kirkpatrick, and B.~Schlein, A central limit theorem in many-body quantum dynamics, Comm. Math. Phys., 321(2):371--417 (2013)

\bibitem{BOS}
N.~Benedikter, G.~de~Oliveira, and B.~Schlein, Quantitative derivation of the {G}ross-{P}itaevskii equation, Comm. Pure Appl. Math., 68(8):1399--1482, (2015).

\bibitem{BPS}
N.~Benedikter, M.~Porta, and B.~Schlein, Effective evolution equations from quantum dynamics,  Springer Briefs in Mathematical Physics, Springer, (2016).

\bibitem{BCCS_betakleiner1}
C.~Boccato, C.~Brennecke, S.~Cenatiempo, and B.~Schlein, The excitation spectrum of bose gases interacting through singular potentials, arXiv:1704.04819 (2017)

\bibitem{BCS}
C.~Boccato, S.~Cenatiempo, and B.~Schlein, Quantum many-body fluctuations around nonlinear {S}chr\"{o}dinger dynamics, Ann. Henri Poincar\'{e},18(1):113--191 (2017)

\bibitem{BNNS}
C.~Brennecke, P.~T. Nam, M.~Napi\'{o}rkowski, and B.~Schlein, Fluctuations of n-particle quantum dynamics around the nonlinear {S}chr\"odinger equation, Ann. Henri Poincar\'{e} (2018)

\bibitem{BS}
C.~Brennecke and B.~Schlein, Gross--{P}itaevskii dynamics for {B}ose--{E}instein condensates, Anal. PDE, 12(6):1513--1596  (2019).

\bibitem{BSS}
S.~Buchholz, C.~Saffirio, and B.~Schlein, Multivariate central limit theorem in quantum dynamics, J. Stat. Phys., 154(1-2):113--152, 2014.

\bibitem{CLS}
L.~Chen, J.~Lee, and B.~Schlein, Rate of convergence towards {H}artree dynamics, J. Stat. Phys., 144(4):872 -- 903 (2011)

\bibitem{CH_GP}
X.~Chen and J.~Holmer, Correlation structures, many-body scattering processes, and the derivation of the {G}ross-{P}itaevskii hierarchy,  Int. Math. Res. Not. IMRN, (10):3051--3110 (2016)

\bibitem{ES}
A.~Elgart and B.~Schlein, Mean field dynamics of boson stars, Comm. Pure Appl. Math., 60(4):500--545 (2007)

\bibitem{ESY_beta_einhalb}
L.~Erd\H{o}s, B.~Schlein, and H.-T. Yau, Derivation of the cubic non-linear {S}chr\"{o}dinger equation from quantum dynamics of many-body systems, Invent. Math., 167(3):515--614 (2007)

\bibitem{ESY_GP_rig}
L.~Erd\H{o}s, B.~Schlein, and H.-T. Yau, Rigorous derivation of the {G}ross-{P}itaevskii equation with a large interaction potential, J. Amer. Math. Soc., 22(4):1099--1156 (2009)


\bibitem{ESY_GP}
L.~Erd\H{o}s, B.~Schlein, and H.-T. Yau, Derivation of the {G}ross-{P}itaevskii equation for the dynamics of {B}ose-{E}instein condensate, Ann. of Math. (2), 172(1):291--370 (2010)

\bibitem{EY}
L.~Erd\H{o}s and H.-T. Yau, Derivation of the nonlinear {S}chr\"{o}dinger equation from a many-body {C}oulomb system, Adv. Theor. Math. Phys., 5(6):1169--1205 (2001)

\bibitem{FKP}
J.~Fr\"{o}hlich, A.~Knowles, and A.~Pizzo, Atomism and quantization, J. Phys. A, 40(12):3033--3045 (2007)

\bibitem{FKS}
J.~Fr\"{o}hlich, A.~Knowles, and S.~Schwarz, On the mean-field limit of bosons with {C}oulomb two-body interaction, Comm. Math. Phys., 288(3):1023--1059 (2009)

\bibitem{Golse}
F.~Golse,  On the dynamics of large particle systems in the mean field limit, arXiv: 1301.5494, (2013). Lecture notes for a course at the NDNS+Applied Dynamical systems summer school "Marcroscopic and large scale phenomena", Universiteit Twente, Enschede (The Netherlands). 

\bibitem{GM}
M.~Grillakis and M.~Machedon, Pair excitations and the mean field approximation of interacting bosons, {II}, Comm. Partial Differential Equations, 42(1):24--67 (2017)

\bibitem{KSS}
K.~Kirkpatrick, B.~Schlein, and G.~Staffilani, Derivation of the two-dimensional nonlinear {S}chr\"{o}dinger equation from many body quantum dynamics, Amer. J. Math., 133(1):91--130 (2011)

\bibitem{KP}
A.~Knwoles and P.~Pickl, Mean-field dynamics: singular potentials and rate of convergence, Comm. Math. Phys., 298(1):101--138 (2010)

\bibitem{LNSS}
M.~Lewin, P.~T. Nam, S.~Serfaty, and J.~P. Solovej, Bogoliubov spectrum of interacting {B}ose gases, Comm. Pure Appl. Math., 68(3):413--471 (2015)

\bibitem{NM_betaeindrittel}
P.~T. Nam and M.~Napi\'{o}rkowski, Bogoliubov correction to the mean-field dynamics of interacting bosons, Adv. Theor. Math. Phys., 21(3):683--738 (2017)

\bibitem{NM_betaeinhalb}
P.~T. Nam and M.~Napi\'{o}rkowski, A note on the validity of {B}ogoliubov correction to mean-field dynamics, J. Math. Pures Appl. (9), 108(5):662--688  (2017)

\bibitem{P_GP_1}
P.~Pickl, Derivation of the time dependent {G}ross-{P}itaevskii equation without positivity condition on the interaction, J. Stat. Phys., 140(1):76--89  (2010)

\bibitem{P_extPot}
P.~Pickl, Derivation of the time dependent {G}ross-{P}itaevskii equation with external fields, Rev. Math. Phys., 27(1):1550003, 45 (2015)

\bibitem{RS_CLT}
S.~Rademacher and B.~Schlein, Central limit theorem for {B}ose-{E}instein condensates, J. Math. Phys., 60(7):071902, 18 (2019)

\bibitem{Schlein}
B.~Schlein, Derivation of effective equations from microscopic quantum dynamics, arXiv:0807.4307 , (2008). Lecture notes for a course at ETH Zurich. 

\bibitem{Spohn_Kinetic}
H.~Spohn, Kinetic equations from {H}amiltonian dynamics: {M}arkovian limits, Rev. Modern Phys., 52(3):569--615 (1980)





\end{thebibliography}
\end{document}